\def\draft{1}
\def\doubleblind{0}
\documentclass[11pt,letterpaper]{article}

\usepackage[margin=1in]{geometry}
\usepackage[utf8]{inputenc}
\usepackage{amsthm}
\usepackage{amsthm, amsmath, amssymb, mathtools}
\usepackage{bm,bbm}
\usepackage{graphicx}
\usepackage{color,xcolor}
\usepackage{paralist,enumitem}
\usepackage{mathrsfs}
\usepackage[normalem]{ulem}
\usepackage{tabularx}
\usepackage{tikz}
\usepackage{caption,comment}
\usepackage{algorithmicx}
\usepackage{algorithm}
\usepackage{algpseudocode}
\newcounter{algsubstate}
\renewcommand{\thealgsubstate}{\alph{algsubstate}}

\algnewcommand\algorithmicinput{\textbf{Input:}}
\algnewcommand\Input{\item[\algorithmicinput]}

\algnewcommand\algorithmicoutput{\textbf{Output:}}
\algnewcommand\Output{\item[\algorithmicoutput]}

\algnewcommand\algorithmicgoal{\textbf{Goal:}}
\algnewcommand\Goal{\item[\algorithmicgoal]}

\newcommand{\alglinenoNew}[1]{\newcounter{ALG@line@#1}}
\newcommand{\alglinenoPop}[1]{\setcounter{ALG@line}{\value{ALG@line@#1}}}
\newcommand{\alglinenoPush}[1]{\setcounter{ALG@line@#1}{\value{ALG@line}}}

\newcommand{\blind}[2]{{\ifnum\draft=1\color{purple}\fi \ifnum\doubleblind=1#2\fi\ifnum\doubleblind=0#1\fi\ifnum\doubleblind=2$\{$ #1 $\vert$ #2 $\}$\fi}}

\makeatletter
\newcommand{\algmargin}{\the\ALG@thistlm}
\algnewcommand{\parState}[1]{\State%
  \parbox[t]{\dimexpr\linewidth-\algmargin}{\strut #1\strut}}

\usepackage[colorlinks=true, allcolors=blue]{hyperref}
\usepackage[capitalize,noabbrev,nameinlink]{cleveref}

\usepackage[style=alphabetic, backend=biber, minalphanames=3, maxalphanames=4, maxbibnames=99, maxcitenames=4]{biblatex}

\usepackage{amsthm}
\usepackage{thmtools,thm-restate}

\numberwithin{equation}{section}
\declaretheoremstyle[bodyfont=\it,qed=\qedsymbol]{noproofstyle}

\declaretheorem[name=Observation,numbered=no]{observation*}

\declaretheorem[numberlike=equation]{theorem}

\declaretheorem[name=Theorem,numbered=no]{theorem*}

\declaretheorem[numberlike=equation]{lemma}
\declaretheorem[name=Lemma,numbered=no]{lemma*}

\declaretheorem[numberlike=equation]{corollary}
\declaretheorem[name=Corollary,numbered=no]{corollary*}

\declaretheorem[numberlike=equation]{proposition}
\declaretheorem[name=Proposition,numbered=no]{proposition*}

\declaretheorem[numberlike=equation]{claim}
\declaretheorem[name=Claim,numbered=no]{claim*}

\declaretheorem[name=Conjecture,numbered=no]{conjecture*}

\declaretheorem[name=Question,numbered=no]{question*}

\declaretheoremstyle[bodyfont=\it]{defstyle} 

\declaretheorem[numberlike=equation,style=defstyle]{definition}
\declaretheorem[unnumbered,name=Definition,style=defstyle]{definition*}

\declaretheorem[numberlike=equation,style=defstyle]{example}
\declaretheorem[unnumbered,name=Example,style=defstyle]{example*}

\declaretheorem[unnumbered,name=Notation=defstyle]{notation*}

\declaretheorem[unnumbered,name=Construction,style=defstyle]{construction*}

\declaretheoremstyle[]{rmkstyle}

\crefname{claim}{Claim}{Claims}
\crefname{fact}{Fact}{Facts}

\DeclarePairedDelimiter{\paren}{\lparen}{\rparen}
\DeclarePairedDelimiter{\set}{\lbrace}{\rbrace}

\DeclarePairedDelimiter{\bracket}{[}{]}
\DeclarePairedDelimiter{\card}{\lvert}{\rvert}

\DeclarePairedDelimiter{\abs}{\lvert}{\rvert}

\newcommand{\mcut}{\mathsf{Max}\text{-}\mathsf{CUT}}
\newcommand{\mdcut}{\mathsf{Max}\text{-}\mathsf{DICUT}}
\newcommand{\mkand}{\mathsf{Max}\text{-}k\text{-}\mathsf{AND}}
\newcommand{\dcut}{\mathsf{DICUT}}

\newcommand{\maxval}[1]{\mathsf{maxval}_{#1}}
\newcommand{\val}[2]{\mathsf{val}_{#1}(#2)}

\newcommand{\1}{\mathbbm{1}}

\newcommand{\type}[3]{\mathsf{type}(#1,#2,#3)}

\newcommand{\nbrtype}[4]{\mathsf{nbhdtype}_{#1,#2}^{#3}(#4)}
\newcommand{\ball}[3]{\mathsf{ball}_{#1}^{#2}(#3)}

\newcommand{\DbAllTypesDeg}[3]{\mathbf{Typ}_{#1}^{#2,#3}}
\newcommand{\NDbTypes}[3]{\mathbf{N}_{#1}^{#2,#3}}

\newcommand{\Dist}[1]{\mathbf{Dists}(#1)}
\newcommand{\EmpDist}[2]{\mathsf{EmpDist}_{#1}(#2)}

\newcommand{\Unif}[1]{\mathsf{Unif}(#1)}
\newcommand{\NoReplace}[2]{\mathsf{NoReplace}_{#1}(#2)}

\newcommand{\degout}[2]{\mathsf{outdeg}_{#1}(#2)}
\newcommand{\degin}[2]{\mathsf{indeg}_{#1}(#2)}
\renewcommand{\deg}[2]{\mathsf{deg}_{#1}(#2)}
\newcommand{\verts}[1]{\mathsf{ends}(#1)} 

\newcommand{\mult}[2]{\mathsf{mult}_{#1}(#2)}
\newcommand{\toset}[1]{\mathsf{set}(#1)}
\newcommand{\vecsigma}{\boldsymbol{\sigma}}
\newcommand{\vecT}{\boldsymbol{T}}
\newcommand{\EdgeDist}[3]{\mathsf{EdgeNbhdTypeDist}_{#1;#2}^{#3}}

\newcommand{\Orderings}[1]{\mathbf{Ords}(#1)}

\newcommand{\tvdist}[2]{\mathsf{tvdist}(#1,#2)}

\newcommand{\Local}{\mathtt{Local}}

\newcommand{\degs}{\mathtt{degs}}

\usepackage{singer-macros}

\allowdisplaybreaks

\title{Streaming Algorithms via Local Algorithms \\ for Maximum Directed Cut}

\ifnum\doubleblind=0
\author{Raghuvansh R. Saxena\thanks{Tata Institute of Fundamental Research, Mumbai, Maharashtra, India. Email: \texttt{raghuvansh.saxena@gmail.com}} 
\and Noah G. Singer\thanks{Department of Computer Science, Carnegie Mellon University, Pittsburgh, PA, USA. Email: \texttt{ngsinger@cs.cmu.edu}.} 
\and Madhu Sudan\thanks{School of Engineering and Applied Sciences, Harvard University, Cambridge, Massachusetts, USA. Email: \texttt{madhu@cs.harvard.edu}.}
\and Santhoshini Velusamy\thanks{Toyota Technological Institute, Chicago, Illinois, USA. Email: \texttt{santhoshini@ttic.edu}.}}
\fi
\ifnum\doubleblind=1
\author{Anonymous authors}
\fi

\date{}
\begin{document}
\maketitle


\begin{abstract}
We explore the use of local algorithms in the design of streaming algorithms for the Maximum Directed Cut problem. Specifically, building on the local algorithm of \textcite{BFSS15,CLS17}, we develop streaming algorithms for both adversarially and randomly ordered streams that approximate the value of maximum directed cut in bounded-degree graphs. In $n$-vertex graphs, for adversarially ordered streams, our algorithm uses $O(n^{1-\Omega(1)})$ (sub-linear) space and for randomly ordered streams, our algorithm uses logarithmic space. Moreover, both algorithms require only one pass over the input stream. With a constant number of passes, we give a logarithmic-space algorithm which works even on graphs with unbounded degree on adversarially ordered streams. Our algorithms achieve any fixed constant approximation factor less than $\frac12$. In the single-pass setting, this is tight: known lower bounds show that obtaining any constant approximation factor greater than $\frac12$ is impossible without using linear space in adversarially ordered streams \textcite{KK19} and $\Omega(\sqrt{n})$ space in randomly ordered streams, even on bounded degree graphs \textcite{KKS15}.

In terms of techniques, our algorithms partition the vertices into a small number of different types based on the structure of their local neighborhood, ensuring that each type carries enough information about the structure to approximately simulate the local algorithm on a vertex with that type. We then develop tools to accurately estimate the frequency of each type. This allows us to simulate an execution of the local algorithm on all vertices, and thereby approximate the value of the maximum directed cut.
\end{abstract}

\color{black}

\newpage

\pagenumbering{arabic} 


\section{Introduction}

We give almost $1/2$-appproximation algorithms solving the \emph{maximum directed cut} ($\mdcut$) problem in graphs in a variety of streaming settings, by appealing to local algorithms achieving similar approximations. We describe the problem and settings in more detail below before turning to the results and techniques.

\subsection{The $\mdcut$ Problem and its Significance}

The {\em maximum directed cut} problem ($\mdcut$) is the problem of estimating the value of the maximum directed cut in an input graph $G$. Here, a (directed) cut is a subset $S$ of the vertices and the value of the cut is the fraction of edges $(u, v)$ in the graph satisfying $u \in S$ and $v \notin S$. We denote the value of the largest cut by $\maxval{G}$ and say that an algorithm produces an $\alpha$-approximation if it guarantees to output a value at least $\alpha \cdot \maxval{G}$ (and at most $\maxval{G}$) on all graphs $G$.

Besides being a central problem in its own right, $\mdcut$ is also significant as it is an example of a \emph{Constraint Satisfaction Problem (CSP)}. In a constraint satisfaction problem, there is a set of variables and a set of constraints over these variables, and the goal is find out the maximum number of constraints that can be satisfied by an assignment to the variables. CSPs form an infinite set of problems that often capture many natural settings, and have received considerable attention in both streaming~\cite{KK15,KKS15,GVV17,KKSV17,BDV18,KK19,AKSY20,AN21,CGS+22-linear-space,BHP+22,CKP+23,CGS+22-monarchy,CGS+22-linear-space,KP22,SSSV23-random-ordering,SSSV23-dicut,KPSY23,Sin23-kand,SSV24,KPV24} and non-streaming settings~\cite[among many others]{Sch78,FV98,Aus07,Rag08,Aus10,Kho10,Mos10,Bul17,Zhu20}. For streaming settings, the $\mdcut$ problem has emerged as a leader on the algorithmic front, with almost all algorithms being developed first for the $\mdcut$ problem before being extended to other CSPs. In fact, speaking in broad strokes, CSPs form an infinite class of problems, in almost all contexts they tend to have a finite classification and in particular there is a finite set of algorithms that essentially cover the entire class. Thus any algorithm that works well for any CSP problem offers hope for the entire class; and the $\mdcut$ problem has proven to be the most suitable for algorithmic developments and insights.

\subsection{Streaming Algorithms for $\mdcut$} 

In the streaming model, the input graph $G$ is presented to the algorithm as a stream of edges, whose goal is to make one or passes over this stream and produce an $\alpha$-approximation using as little memory as possible. (In this paper, all algorithms are allowed to toss random coins and need to succeed with probability $99\%$ over the choice of random coins for every input.) It is easy to see that storing $O(n)$ randomly chosen edges suffices to produce a $\paren*{ 1- \epsilon }$-approximation, and this can be done in a single pass and $\tilde{O}(n)$ memory. However, a linear-sized memory is very often unaffordable and thus, research has mostly focused on investigating the power of sublinear algorithms for $\mdcut$. 

On this front, the work of \cite{KK19} showed that $\Omega(n)$ space is needed to achieve any approximation better than $1/2$ when the edges are ordered adversarially, even on bounded degree graphs. It is believed that the proof techniques also extend to randomly ordered streams although the best bound known is a memory lower bound of $\Omega(\sqrt{n})$ in the early work of \cite{KKS15}. Thus, the best approximation factor one can hope for with sublinear space is $1/2$ and here, the $\mdcut$ problem does admit many non-trivial results.

The first non-trivial streaming algorithm for approximating $\mdcut$ was due to \cite{GVV17} who gave a single-pass that uses logarithmic space and produces a $2/5$ approximation. Subsequently, \cite{CGV20} improved this to a $4/9$ approximation and also proved that a better approximation is impossible without using at least $\Omega(\sqrt{n})$ space if the edges are ordered adversarially. Subsequently, \cite{SSSV23-random-ordering,SSSV23-dicut} showed that even better approximations (of around $0.485$) can be produced if algorithms are allowed $\tilde{O}(\sqrt{n})$ space. Moreover, the space bound improves to logarithmic if the edges in the graph as assumed to arrive in a random order.

The above works, summarized in \cref{tab:algs}, have led to improvements for a wide class of constraint satisfaction problems. For example, \cite{CGSV24} extended the work of \cite{CGV20} mentioned above to cover all possible CSPs and \cite{Sin23-kand} showed that the algorithms in \cite{SSSV23-random-ordering,SSSV23-dicut} can also be extended to a richer class. This again demonstrates the significance of algorithms for $\mdcut$ when it comes to designing algorithms for general CSPs. However, the following question remains open: Can sublinear space streaming algorithms achieve $(1/2-\epsilon)$-approximations for the $\mdcut$ problem, for \emph{every} $\epsilon > 0$? While we do not resolve this question here, we show many special cases, in particular including most used/sufficient for existing lower bounds, do have such an approximation algorithm in sublinear space.

\begin{table}
    \centering
    \begin{tabular}{c|c|c|c|c|c}
         \textbf{Citation} & \textbf{Approx. factor} & \textbf{Input order} & \textbf{Space} & \textbf{Passes} & \textbf{Bounded-degree?} \\ \hline
         Folklore & $1-\epsilon$ & Adversarial & $\tilde{O}(n)$ & $1$ & No \\
         \cite{GVV17} & $2/5-\epsilon$ & Adversarial & $O(\log n)$ & $1$ & No \\
         \cite{CGV20} & $4/9-\epsilon$ & Adversarial & $O(\log n)$ & $1$ & No \\
         \cite{SSSV23-random-ordering} & $0.485$ & Adversarial & $\tilde{O}(\sqrt n)$ & $1$ & Yes \\
         \cite{SSSV23-random-ordering} & $0.485$ & Random & $O(\log n)$ & $1$ & No \\
         \cite{SSSV23-random-ordering} & $0.485$ & Adversarial & $O(\log n)$ & $2$ & No \\
         \cite{SSSV23-dicut} & $0.485$ & Adversarial & $\tilde{O}(\sqrt n)$ & $1$ & No \\
         This work & $1/2-\epsilon$ & Adversarial & $o(n)$ & $1$ & Yes \\
         This work & $1/2-\epsilon$ & Random & $O(\log n)$ & $1$ & Yes \\
         This work & $1/2-\epsilon$ & Adversarial & $O(\log n)$ & $O(1)$ & No
    \end{tabular}
    \caption{A table of known streaming algorithms for $\mdcut$. Some lower bounds (that hold even for bounded degree graphs) are also known: $(4/9+\epsilon)$-approximation in single-pass adversarial-ordering streams requires $\Omega(\sqrt n)$ space \cite{CGV20}; $(1/2+\epsilon)$-approximation in single-pass adversarial-ordering streams requires $\Omega(n)$ space \cite{KK19}; $(1/2+\epsilon)$-approximation in single-pass random-ordering streaming requires $\Omega(\sqrt n)$ space \cite{KKS15}; and $(1-\epsilon)$-approximation over adversarially-ordered streams requires either $n^{\Omega(1)}$ space or $\Omega(1/\epsilon)$ passes \cite{AN21}. We conjecture that the $(1/2+\epsilon)$ lower bounds can be simultaneously generalized, to show that $(1/2+\epsilon)$-approximation in single-pass random-ordering streaming requires $\Omega(n)$ space.}
    \label{tab:algs}
\end{table}

\subsection{Our Results}

As mentioned above, we achieve optimal $(1/2-\epsilon)$-approximation for the $\mdcut$ problem in various restricted settings. 

\paragraph{Bounded degree graphs with adversarially ordered edges.} The first variant we consider is when the input graph is assumed to have a bounded degree, i.e., the degree of any vertex is assumed to be at most a pre-specified constant. While the bounded degree setting does make it easier to effect algorithmic improvements, it also tends to be a strong predictor of general results, in that the algorithms can often be extended (possibly with many complications in algorithm design as well as analysis) to the general setting. As an example, the $\tilde{O}(\sqrt n)$-space $\mdcut$ algorithm of \cite{SSSV23-random-ordering} for bounded-degree graphs was extended to general graphs in \cite{SSSV23-dicut}. Conversely, as far as we are aware, all known streaming lower bounds for $\mdcut$ (i.e., \cite{KKS15,KK19}) are based on random sparse graphs where the maximum degree of any vertex is at most constant (or logarithmic). Together, these past algorithms and lower bounds motivate the search for algorithms in this setting. For this setting, we show that\footnote{Throughout, the $O(\cdot)$ and $\Omega(\cdot)$ notations hide constants depending on $\epsilon$ and, if present, $D$.}:

\begin{theorem}[Adversarial-order algorithm for bounded degree graphs]
\label{thm:single-pass:adversarial-order}
For every $D \in \BN$ and $\epsilon>0$, there is a streaming algorithm which $(1/2-\epsilon)$-approximates the $\mdcut$ value of an $n$-vertex graph with maximum degree at most $D$ in $O(n^{1-\Omega(1)})$ space using a single, adversarially-ordered pass over the list of edges.
\end{theorem}

\paragraph{Bounded degree graphs with randomly ordered edges.} The random order streaming model is now a central model in streaming literature, and understanding it is an important quest of its own~\cite[etc.]{KKS14,MMPS17,PS18,FHM+19,AB21,Ber23,JW23}. The difference between the random order model and the adversarially ordered model is that in the random order model, the edges of the input graph\footnote{The input graph is still worst-case, just that its edges are presented in a random order.} are presented in a uniformly random order to the streaming algorithm, and the algorithm has to do well on most orders. \cite{SSSV23-random-ordering} show that it is possible to produce an approximation of around $0.485$ of $\mdcut$ using only logarithmic space, thereby going past the $4/9$ lower bound that holds for adversarially ordered streams \cite{CGV20}. However, there is still a gap between the approximation guarantee of the best known algorithm (around $0.485$) and the best known lower bound of $1/2 + \epsilon$ shown in \cite{KKS15}, that also holds for bounded degree graphs. Our work here closes this gap for bounded degree graphs. 

\begin{theorem}[Random-order algorithm for bounded degree graphs]
\label{thm:single-pass:random-order}
For every $D \in \BN$ and $\epsilon>0$, there is a streaming algorithm which $(1/2-\epsilon)$-approximates the $\mdcut$ value of an $n$-vertex graph with maximum degree at most $D$ in $O(\log n)$ space using a single, randomly-ordered pass over the list of edges.
\end{theorem}

\paragraph{Multi-pass algorithms.} Finally, we consider the setting where the streaming algorithm is allowed to make multiples passes over the (adversarially ordered) input stream. The multi-pass setting has also been widely studied~\cite{GM08,DFR10,KMM12,GO16,KT17,Ass17,BC17,BCK+18,ACK19,AKSY20,CKP+21b,CKP+21a,CKP+23,BGL+24} in the streaming literature. In the context of $\mdcut$, the most relevant work is that of \cite{SSSV23-random-ordering} that produces an approximation of around $0.485$ using only two passes and logarithmic space, thereby surpassing the $4/9$ lower bound that holds for single-pass algorithms \cite{CGV20}. We show that a constant number of additional passes can push the approximation factor arbitrarily close to $1/2$. It is believed that, for the related problem of $\mcut$ (which implies it\footnote{An algorithm for solving $\mdcut$ implies an algorithm for $\mcut$, simply replace every edge by two edges, one in each direction.} for $\mdcut$), even algorithms with logarithmically many passes need polynomial space to produce an approximation larger that $1/2$ \cite{CKP+23}.

\begin{theorem}[Multi-pass algorithm]
\label{thm:multi-pass}
For every $\epsilon>0$, there is a streaming algorithm which $(1/2-\epsilon)$-approximates the $\mdcut$ value of an arbitrary $n$-vertex graph in $O(\log n)$ space using $O(1/\epsilon)$ adversarially-ordered passes over the list of edges.
\end{theorem}

\section{Our techniques}

Our algorithms in the streaming setting build on a \emph{local computation} algorithm from the work of \cite{BFSS15,CLS17}, and we begin by discussing this algorithm. This algorithm is based on the submodularity of the $\dcut$ predicate and is local in the following precise sense. Suppose $G$ is a $k$-colored graph. Every vertex has a \emph{(colored) $k$-neighborhood}, which contains the set of all vertices reachable from $v$ by undirected paths of length at most $k$, the induced subgraph of $G$ on these vertices, and these vertices' colors. Then each vertex $v$ deterministically receives a fractional assignment $x_v \in [0,1]$ depending only on its $k$-neighborhood, and this assignment is guaranteed to have value at least $\frac12$ of the optimal value, i.e.,
\begin{equation}\label{eq:half-approx}
    \frac1m \sum_{(u,v) \in E} x_u (1-x_v) \geq \frac12 \maxval{G}.
\end{equation}
Moreover, $x_v$ depends only on the \emph{isomorphism class} of the ``colored neighborhood'' of $v$, i.e., it is invariant to relabeling vertices.\footnote{We mention a different approach to local $(1/2-\epsilon)$-approximations for $\mdcut$: Kuhn, Moscibroda, and Wattenhofer~\cite{KMW06} developed a local algorithm for $(1-\epsilon)$-approximating the value of packing-and-covering linear programs (see also the thesis of Kuhn~\cite{Kuh05}), and there is a well-known linear programming relaxation of $\mdcut$ (see e.g. \cite{Tre98-alg}) that is a packing-and-covering LP and is half-integral. Combining them could give a local algorithm for $(1/2-\epsilon)$-approximating $\mdcut$, although we do not investigate this and make no formal claim in this regard.}

The main idea behind our results is to import the \cite{BFSS15,CLS17} algorithm for $\mdcut$ into the streaming paradigm to get $(1/2-\epsilon)$-approximations. (A similar theme of translating local algorithms to sublinear algorithms was explored in the work of \cite{PR07} for problems like minimum vertex cover.) We can essentially assume without loss of generality that our input graphs are $k$-colored for some large constant $k = O(1/\epsilon)$ (using a random $k$-coloring and discarding improperly colored edges). Given this coloring, we aim to \emph{sample a random edge $(u,v)$ in $G$ and output the $k$-neighborhoods of its endpoints}, since these $k$-neighborhoods determine $x_u$ and $x_v$ and therefore $x_u(1-x_v)$; we then use multiple random samples to estimate the quantity on the right-hand side of \Cref{eq:half-approx} and produce a $1/2$-approximation.

We now explain the high-level ideas behind the algorithms in \cref{thm:single-pass:adversarial-order,thm:single-pass:random-order,thm:multi-pass}. Our first two algorithms use the \cite{BFSS15,CLS17} algorithm as a black box, while the third relies on building a certain ``robust'' version of their algorithm.

\paragraph{Bounded degree graphs with adversarially ordered edges (\cref{thm:single-pass:adversarial-order}).} We define the $k$-neighborhood of an \emph{edge} $e$ as the induced subgraph on the set of vertices of distance at most $k$ from either endpoint of $e$ (i.e., as the induced subgraph on the union of the $k$-neighborhoods of its endpoints). The \emph{type} of an edge is the isomorphism class of this neighborhood.

The starting point of this algorithm is the observation that, in the bounded-degree setting, each vertex has only a constant number of other vertices in its constant-radius neighborhood. Thus, an edge has only a constant number of possible types and we can aim to estimate the distribution of the type of a random edge via a streaming algorithm.

Our procedure looks roughly as follows. Before the stream, we sample a large (but $o(n)$-sized) set $S$ of vertices. Then, we use our adversarial-order pass to record two kinds of information about the input graph:
\begin{enumerate}
    \item $G[S]$, the induced subgraph of $G$ on the vertices $S$.
    \item $\{\degin{G}{v},\degout{G}{v}\}_{v \in S}$, the in-degree and out-degree of each vertex in $S$.
\end{enumerate}

At this point, it is crucial to distinguish between two notions of an edge's neighborhood: An edge $(u,v) \in E$ has a \emph{true} neighborhood, i.e., its actual $k$-neighborhood in $G$, as well as what we call its \emph{induced} neighborhood, i.e., its neighborhood in the induced subgraph $G[S]$. This neighborhood should be thought of as the subset of the true neighborhood which the streaming algorithm ``sees''. Importantly, the true neighborhoods are fixed and depend only on $G$, while the induced neighborhoods depend also on the set $S$ and therefore on the streaming algorithm's randomness.

The induced neighborhood of an edge might be a strict subset of the true neighborhood: For instance, if a vertex $v$ is in $S$, only some of its neighbors in $G$ might also in $S$. This is where the second kind of information the algorithm records --- the global in- and out-degrees of every vertex in $S$ --- lets us \emph{reject} cases where the induced type does not match the true type. Indeed, we can tell whether a edge $e = (u,v)$'s induced neighborhood is its true neighborhood (equivalently, whether all vertices in $e$'s $k$-neighborhood were included in $S$) by checking whether $u$'s ``induced degree'' (degree in $G[S]$) matches its ``true degree'' (degree in $G$) and the same for $v$, and so on for their neighbors and their neighbors recursively out to depth $k$. 

To analyze this algorithm, imagine performing a ``diagnostic test'' on the (multi)set of edges in the graph in order to estimate the number of type-$T$ edges. Our test on an edge $e$ is: Is $e$'s induced neighborhood is isomorphic to $T$ \underline{and} does $e$'s induced neighborhood match its true neighborhood? This test always rejects non-type-$T$ edges, but also often rejects type-$T$ edges. To estimate the actual number of type-$T$ edges, we need to normalize by the probability that a type-$T$ edge is accepted. (In statistical parlance, the test has no false positives, but does have false negatives, and we normalize by the false negative rate, which is called the \emph{sensitivity}.) This probability is the same for every type-$T$ edge: If $T$ contains $a$ vertices and $S$ contains each vertex independently with probability $p$, this probability is simply $p^a$. So, we take $p \sim n^{-1/a}$ for this probability to be a constant, which is also enough to get good concentration, and the entire algorithm then takes roughly $\tilde{O}(n^{1-1/a})$ space. Similar (but less general) tools were also present in \cite{SSSV23-random-ordering}.

Before continuing to the other results, we state a subtlety that we overlooked in our description above. While our algorithm assumes that the input graph has bounded maximum degree, and therefore at most linearly many edges, the graph could actually have sublinearly many edges. In this case, the graph could have many isolated vertices, and to compensate for this, $S$ needs to contain each \emph{non}isolated vertex with higher probability than $n^{-1/a}$ (it needs to scale with $m$). If the graph has this issue, we sample $S$ to be larger than $n^{1-1/a}$ initially. A priori this might hurt the space bound, but we observe that we actually expend no space for storing isolated vertices in $S$ (they have no incident edges in the induced subgraph $G[S]$!). (To avoid having to write down the description of $S$ itself, we sample it as an $O(1)$-wise independent set of vertices instead of a fully independent set of vertices.) For simplicity, we continue to ignore this subtlety in this overview.

\paragraph{Bounded degree graphs with randomly ordered edges (\cref{thm:single-pass:random-order}).} This algorithm builds on the foregoing one. Just as before, we still assume that the maximum degree of the vertices is bounded, and we can still partition the vertices into a constant number of types and try to compute the frequency of each type. The main challenge is that we are now trying to do this in logarithmic (instead of simply sublinear) space using the fact that the edges are presented randomly.\footnote{Recall that the \cite{CGV20} lower bound implies that getting logarithmic space is impossible if the edges are presented in an adversarial order.}

This means that we can no longer sample a set of vertices of size $O(n^{1-\Omega(1)})$ and consider the induced subgraph. Instead, we sample logarithmically many vertices and ``build'' their neighborhoods as the stream passes. For example, if we are building the neighborhood of vertex $u$ (say), we start looking for edges incident to vertex $i$ and add them to the neighborhood. As soon as we find such an edge $(u, v)$ (say), we additionally start looking (recursively) for edges with vertex $v$ until we reach our fixed (constant) depth $k$. We hope to explore the entire neighborhood of vertex $u$ and thereby compute its type.

We again distinguish between the \emph{true} neighborhood of a vertex (its neighborhood in $G$) and a subset of the neighborhood, which we call the \emph{visible} neighborhood, which is now the result of this greedy neighborhood-building procedure. As in the induced neighborhood from the adversarial-ordering case, this new notion of the visible neighborhood of a vertex is a random subset of its true neighborhood and should be interpreted as the part of the true neighborhood that the streaming algorithm ``sees''. However, one important difference is that in the adversarial-ordering case, the induced neighborhoods depended on the algorithm's randomness (in choosing the set $S$), while in this new random-ordering case, the visible neighborhoods now depend instead on the randomness of the \emph{stream}; indeed, the algorithm itself is actually deterministic.

Similarly to induced neighborhoods in the adversarial-ordering case, the visible neighborhood of a vertex in our random-ordering algorithm is often only a strict subset of its true neighborhood. For instance, suppose our graph has the edges $(1,2),(1,3),(3,4)$, and we are building the neighborhood of vertex~$1$. If $(3,4)$ occurs before $(1,3)$ in the stream, it will ``slip by'' and we will not include it in the neighborhood. If we're looking for vertices of fixed neighborhood type $T$, we will therefore get ``false negatives'' again: Vertices which do have true type $T$ but whose visible types are strict subsets of $T$. In the adversarial-ordering algorithm described above, we also had global $G$-degree information which allowed us to directly reject these false negatives. The key challenge in the random ordering setting is that the algorithm no longer has information allowing us to perform such rejections.

But it is not necessary to directly reject false negatives. Instead, we can account for their effect indirectly, in the following sense. Each vertex $v$ in $G$ can produce a variety of possible visible types depending on the randomness of the stream. The key insight is that the marginal distribution of $v$'s visible type depends only on the true type of $v$! So, while we do not directly know $v$'s true type, we do get a sample from a distribution depending only on $v$'s true type. Thus, we can use Bayes' rule to estimate the frequencies of true types in $G$ from the frequencies of visible types in our sample. (The astute reader may notice that vertices' visible types might not be independent. However, a set of vertices whose $k$-neighborhoods are all mutually disjoint will indeed have mutually independent visible types, and in a bounded-degree graph, a random set of vertices is likely to be spread out in this sense.) We mention that similar tools for measuring distributions were also used in \cite{MMPS17}. (In fact, \cite{MMPS17} contains a black-box algorithm for estimating the neighborhood-type distribution of random vertices. But as mentioned above, our actual algorithm needs to understand neighborhoods of \emph{edges}, not vertices. Our main technical contribution here is thus a modification of the \cite{MMPS17} algorithm to this setting.)

\paragraph{Remarks on the single-pass algorithms.}
\begin{itemize}
    \item Our single-pass streaming algorithms measure `strong'  information about the input graph: The distribution of the isomorphism type of the induced subgraph on the $k$-neighborhood around a random vertex --- or, rather, since we need information about edge types, the induced subgraph on the union of $k$-neighborhoods of the endpoints of a random edge. However, we emphasize that the local algorithm only uses much `weaker' information about a random edge $(u,v)$; for instance, it does not need to know how the ball around $u$ intersects with the ball on $v$, nor does it need to know about neighbors $w$ of $u$ or $v$ whose color is larger than $u$'s and $v$'s colors.
    \item The single-pass streaming algorithms rely crucially on assuming that the maximum degree of each vertex is bounded. Indeed, the premise of both algorithms is to somehow measure the distribution of ``true types'' in the graphs, meaning the distribution of the isomorphism class of the $k$-neighborhood of a uniformly random vertex (or, actually, of a random edge). In an unbounded-degree graph, this distribution may not even have constant-sized support (and therefore may not be outputtable by a low-memory algorithm)! Note also that the adversarial-ordering algorithm uses $\sim n^{1-1/a}$ space to measure the frequency of a type $T$ with $a$ vertices, and if $a$ were polynomial is $n$, then the algorithm uses linear space.
\end{itemize}

\paragraph{Multi-pass algorithms (\cref{thm:multi-pass}).} If a graph is promised to be bounded degree, then there is a simple and deterministic way to measure the ``true types'' of vertices using logarithmic space and constantly many passes: Given a starting vertex $v$, use the first pass to query the neighbors of $v$; the second pass to query these neighbors' neighbors, and so on. This procedure is guaranteed to exactly compute the true type of $v$ and uses only logarithmic space if the graph has bounded-degree.

But our multi-pass algorithm avoids making a bounded-degree assumption about the input graph. Thus, its flavor is quite different from the previous two single-pass algorithms, because unlike those algorithms, we can no longer afford to store the entire true types of vertices, which may be arbitrarily large. In particular, we can no longer rely on the \cite{BFSS15,CLS17} algorithm as a black box.

Recall that the \cite{BFSS15,CLS17} algorithm uses the entire $k$-neighborhood of a vertex $v$ to produce a fractional assignment $x_v \in [0,1]$ for $v$. Informally, we ``robustify'' the \cite{BFSS15,CLS17} algorithm to produce an estimated fractional assignment for $v$ based on \emph{random subsampling} of its $k$-neighborhood. We show that this estimate is likely close to $x_v$. The subsampled graph is bounded-degree (in fact, $O(1)$-regular) and therefore our algorithm uses only logarithmic space.

We now give a brief overview of the \cite{BFSS15,CLS17} and our modification. To compute the (fractional) assignment $x_v \in [0,1]$ for a vertex $v$ in $G$, the \cite{BFSS15,CLS17} algorithm only uses a few quantities:
\begin{enumerate}
    \item The sum of fractional assignments $x_v$ for \emph{lower}-colored neighbors of $v$.
    \item Simple degree statistics of $v$: how many in- and out-edges in $G$ does $v$ have to lower- and higher-colored neighbors. (In the base case, color-$1$ vertices, the assignment depends only on these statistics.)
\end{enumerate}
Note that only the first item involves recursive applications of the algorithm. Also, the recursion only has depth $k$, since we recurse only on lower-colored neighbors and there are $k$ total colors.

Our robust local version of the \cite{BFSS15,CLS17} algorithm should be interpreted as a random truncation/pruning of that algorithm's recursion tree so that every vertex makes a constant number of recursive calls. Indeed, we \emph{estimate} the sum in the first item above by randomly sampling a constant-sized subset of $v$'s lower-color neighbors and only recursing on these neighbors.

To make this more precise, our robust local algorithm looks like the following. For every vertex $v \in G$, we recursively define a distribution $Y_v$. To sample from $Y_v$:
\begin{enumerate}
    \item Sample $D$ random and independent lower-color neighbors $u_1,\ldots,u_D$ of $v$, where $D$ is a large constant depending on $\epsilon$. (Note: For convenience, the neighbors $u_1,\ldots,u_D$ are sampled with replacement. So even if $v$ has lower degree than $D$, we still sample $D$ neighbors.)
    \item For each $i \in [D]$, sample an estimate $y_i \sim Y_{u_i}$. (If $u_i = u_j$, $y_i$ and $y_j$ are still sampled independently.)
    \item Use the sum $\sum_{i=1}^D y_i$ together with the degree statistics to compute an estimate of $x_v$ via a similar procedure to the \cite{BFSS15,CLS17} algorithm.
\end{enumerate}
Note that each $y_i$ can deviate from $x_{u_i}$; our new estimate for $x_v$ combines these $y_i$'s and may deviate further from $x_v$ if the errors compound in the right way. Further, we must track the error probabilities, since each of the $y_i$'s might deviate too much from $x_{u_i}$, and the sample $u_1,\ldots,u_D$ itself might not be representative.

The heart of our multi-pass algorithm is an analysis which manages these compounding errors. Once this analysis is complete, the resulting robust local algorithm is simple to implement in the streaming setting with $O(k)$ passes: In each pass, we start with a ``layer'' of vertices, measure their degree statistics, and sample $D$ random neighbors for each, which in turn form the next layer.

\subsection{Future directions}

The most immediate future direction is to try to extend our single-pass bounded-degree algorithms (\cref{thm:single-pass:adversarial-order,thm:single-pass:random-order}) to the setting of general graphs (i.e., without the bounded-degree assumption). A natural starting point would be the machinery we develop for the \emph{multi-}pass algorithm (\cref{thm:multi-pass}) which eliminates the bounded-degree assumption there. However, there appear to be significant technical challenges, fundamentally because the number of underlying isomorphism classes of vertices' neighborhoods no longer has constant size. We essentially get around this in the multi-pass case by ``randomly truncating'' the neighborhoods of high-degree vertices, so that we recurse only a random, constant-sized sample of their full neighborhoods. It is not clear how to combine this technique with the mechanisms we develop in the single-pass setting. 

Another interesting question is whether the algorithms developed in this paper could be adapted into \emph{quantum} streaming algorithms, just as recent work of \cite{KPV24} adapted the algorithm of \cite{SSSV23-dicut}. Finally, it would be very interesting to design local algorithms for other CSPs aside from $\mdcut$ (such as $\mkand$) and to generalize our streaming results to these problems. 

\subsection*{Outline}

In \cref{sec:prelim}, we write some notations for and basic facts about multisets, graphs, probability, and total variation distance which we employ in the paper. In \cref{sec:nbr-types}, we develop a notion of the ``neighborhood type'' of an edge and state a key connection (\cref{thm:types-to-dicut}) between the neighborhood type of a uniformly random edge in a graph and approximations of the $\mdcut$ value of the same graph. We use this connection in \cref{sec:adv-ord,sec:rand-ord} to design algorithms for the single-pass unbounded-degree adversarial-ordering and random-ordering settings, thereby proving \cref{thm:single-pass:adversarial-order,thm:single-pass:random-order}, respectively. These algorithms are both based on implementing estimators for the ``edge neighborhood-type distribution'' defined in the previous section. Finally, in \cref{sec:multipass}, we describe and analyze our multipass algorithm (\cref{thm:multi-pass}), which avoids making any random-ordering assumptions.


\section{Preliminaries}\label{sec:prelim}

For $k \in \BN$, $[k]$ denotes the natural numbers between $1$ and $k$ inclusive. Given a function $f : S \to T$ and a subset $U \subseteq S$, $f\vert_U:U \to T$ denotes the restriction of $f$ to $U$.

For a finite set $S$, $\Dist{S}$ denotes the set of all probability distributions over $S$; for $\CD \in \Dist{S}$, $\CD(S)$ denotes $\Pr_{s' \sim \CD}[s'=s]$.

\subsection{Multisets}

A \emph{multiset} may contain multiple (finitely many) copies of elements; for a multiset $S$, the \emph{multiplicity} of $s$, denoted $\mult{S}{s}$, is the number of copies of $s$ in $S$; $\toset{S}$ is the conversion of $S$ to a set (by forgetting the multiplicity information); and $|S| = \sum_{s \in \toset{S}} \mult{S}{s}$ is the total number of elements of $S$.\footnote{A multiset $S$ is formally a pair $(\toset{S},\mult{S}{\cdot} : \toset{S} \to \BN)$.} If $S$ and $T$ are multisets, we write $S \subseteq T$ to denote that with $\toset{S} \subseteq \toset{T}$ and $\mult{S}{s} \leq \mult{T}{s}$ for all $s \in \toset{S}$. If $S$ is a multiset, then $\Unif{S} \in \Dist{\toset{S}}$ is the distribution over $\toset{S}$ which takes value $s \in \toset{S}$ with $\frac1{|S|} \mult{S}{s}$.

For a multiset $S$, $\Orderings{S}$ denotes the set of \emph{orderings} on $S$, i.e., functions $\sigma : [|S|] \to \toset{S}$ such that $|\{i \in [n] : \sigma(i) = s\}| = \mult{S}{s}$ for every $s \in \toset{S}$.

\subsection{Graphs}

In this paper, a \emph{graph} is a \emph{directed} \emph{multi}graph, i.e., $G = (V,E)$ for a finite set of \emph{vertices} $V$ and a \emph{multiset} $E \subset V \times V \setminus \{(v,v) : v \in V\}$ of \emph{edges}. For an edge $e = (u,v)$, we let $\verts{e} := \{u,v\}$ denote the set of $e$'s endpoints.

Let $G = (V,E)$ be a graph. The \emph{in-degree} and \emph{out-degree} of $v$ are \[ \degin{G}{v} := \sum_{u \in V} \mult{G}{u,v} \hspace{1.5cm} \degout{G}{v} := \sum_{u \in V} \mult{G}{v,u}, \] respectively, and the \emph{total degree} (or just \emph{degree}) of $v$ is \[ \deg{G}{v} := \sum_{u \in V} (\mult{G}{u,v}+\mult{G}{v,u}) = \degout{G}{v} + \degin{G}{v}. \] The \emph{maximum degree} of $G$ is the maximum degree of any vertex. $G$ is \emph{$D$-bounded} if its maximum degree is at most $D$.

Let $G = (V,E)$ be a graph and let $x : V \to \{0,1\}$ be a labeling of $G$'s vertices by Boolean values. The \emph{$\mdcut$ value} of $x$ on $G$ is 
\begin{equation}
\label{eq:val}
\val{G}{x} := \frac1{|E|} \sum_{(u,v) \in E} x(u) (1-x(v)) ,
\end{equation}
(the sum is counted with multiplicity). (An edge $(u,v)$ is \emph{satisfied} by $x$ if $x(u) (1-x(v)) = 1$, i.e., if $x(u) = 1$ and $x(v) = 0$. In this sense, the $\mathsf{DICUT}$ value is the fraction of satisfied edges.)

Further, we use \cref{eq:val} to define $\val{G}{x}$ for ``fractional'' assignments $x : V \to [0,1]$ in the natural way. Observe that $\val{G}{x}$ equals the expected value of the Boolean assignment that assigns $v$ to $1$ w.p. $x(v)$ and $0$ w.p. $1-x(v)$. Hence by averaging:

\begin{proposition}\label{eq:prelim:rounding}
    Let $G = (V,E)$ be any graph and $x : V \to [0,1]$ any fractional assignment. Then there exist Boolean assignments $y, z : V \to \{0,1\}$ such that \[ \val{G}{y} \leq \val{G}{x} \leq \val{G}{z}. \]
\end{proposition}

The \emph{$\mdcut$ value} of $G$, without respect to a specific assignment $x$, is 
\begin{equation}
\label{eq:maxval}
\maxval{G} := \max_{x : V \to \{0,1\}} \val{G}{x} .
\end{equation}

Let $G = (V,E)$. The \emph{induced subgraph} of $G$ on a subset of vertices $U \subseteq V$ is the graph $G[U] = (U,E[U])$ where $E[U]$ is the subset of edges in $E$ with both endpoints in $U$ (with multiplicity). We'll need the following simple fact about induced subgraphs:

\begin{proposition}\label{prop:prelim:doubly-induced-subgph}
    Let $G = (V,E)$, $U \subseteq V$, and $W \subseteq U$. Then $(G[U])[W] = G[W]$.
\end{proposition}

Let $G = (V,E)$ and $u,v \in V$. An \emph{path} from $u$ to $v$ is a sequence of vertices $u = w_0,w_1,\ldots,w_{\ell-1},w_\ell=v \in V$ such that for each $i \in [\ell]$, $(w_{i-1},w_i) \in E$ or $(w_i,w_{i-1}) \in E$.\footnote{Note that in this notion of ``paths'', the directions of edges are ignored.} For $v \in V$, we let the ``radius-$\ell$ ball'' around $v$ be \[ \ball{G}{\ell}{v} := \{u \in V : \exists \text{ a path of length}\leq \ell \text{ between }u\text{ and } v\text{ in }G\}. \] We use the following loose bound on the size of these balls:

\begin{proposition}\label{prop:prelim:ball-size}
    For every $\ell,D \in \BN$, $D \geq 2$, if $G=(V,E)$ has maximum degree $D$, then for every $v \in V$, $|\ball{G}{\ell}{v}| \leq 2D^\ell$.
\end{proposition}

\begin{proof}
    The number of paths originating at $v$ is at most $1 + D + \cdots + D^\ell$. Using the geometric sum formula, this quantity equals $\frac{D^{\ell+1}-1}{D-1} \leq \frac{D^{\ell+1}}{D-1} \leq 2D^\ell$ (since $D-1 \geq D/2$).
\end{proof}

A \emph{(proper) $k$-coloring} of $G = (V,E)$ is a function $\chi : V \to [k]$ such that for all $e = (u,v) \in E$, $\chi(u) \neq \chi(v)$. (A coloring is any general function $\chi : V \to [k]$.) 

\begin{definition}[Colored graph]
    A \emph{(properly) $k$-colored graph} is a pair $(G=(V,E),\chi : V \to [k])$, where $G$ is a graph and $\chi$ is a (proper) $k$-coloring of $G$.
\end{definition}

\subsection{Probability}

For two probability distributions $\CX,\CY \in \Dist{S}$, the \emph{total variation distance} between the distributions is
\begin{equation}
\label{def:prelim:tv-dist}
    \tvdist{\CX}{\CY} := \frac12 \sum_{s \in S} |\CX(s) - \CY(s)| \, .
\end{equation} 

We use some standard facts about the distance:

\begin{proposition}\label{prop:prelim:tv-diff}
    Let $\CX,\CY$ be two probability distributions with $\tvdist{\CX}{\CY}\le \epsilon$ supported on a finite set $S$ and let $f : S \to [0,1]$ be any function. Then
    \[
    \left\lvert \Exp_{s \sim \CX}[f(s)] - \Exp_{s \sim \CY}[f(s)]\right\rvert \leq \epsilon.
    \]
\end{proposition}

Given a set $s_1,\ldots,s_t \in S$, we define the \emph{empirical distribution} $\EmpDist{S}{s_1,\ldots,s_t} \in \Dist{S}$ via $\EmpDist{S}{s_1,\ldots,s_t}(s) := \frac1t |\{i \in [t] : s_i = s\}|$.

\begin{proposition}\label{prop:prelim:tv-samples}
    For every finite set $S$ and $\epsilon, \delta > 0$, there exists $t \in \BN$ such that the following holds. Let $\CD$ be any distribution over $S$. Then w.p. $1-\delta$ over $t$ independent samples $s_1,\ldots,s_t \sim \CD$,  $\tvdist{\CD}{\EmpDist{S}{s_1,\ldots,s_t}} \leq \epsilon$.
\end{proposition}

Also, for finite sets $S,T$, a function $\CF : S \to \Dist{T}$, and a distribution $\CD \in \Dist{S}$, $\CF\circ\CD$ denotes the ``composite'' random variable which samples $s \sim \CD$ then outputs a sample from $\CF(s)$.

\begin{proposition}[``Data processing inequality'']\label{prop:prelim:data-proc}
    Let $S, T$ be finite sets and $\CF : S \to \Dist{T}$ any function. Further, let $\CD_1,\CD_2 \in \Dist{S}$. Then \[ \tvdist{\CF\circ\CD_1}{\CF\circ\CD_2} \leq \tvdist{\CD_1}{\CD_2}. \]
\end{proposition}

\begin{proposition}\label{prop:prelim:sample-diff}
    Let $S, T$ be finite sets and $\CF, \CG : S \to \Dist{T}$ any functions. Let $\CD \in \Dist{S}$. Then \[ \tvdist{\CF\circ\CD}{\CG\circ\CD} \leq \Pr_{s \sim \CD}[\CF(s) \neq \CG(s)]. \]
\end{proposition}

Also, for $t \leq |S|$, we define $\NoReplace{t}{S}$ as the distribution over $(\toset{S})^t$ which iteratively samples $s_1,\ldots,s_t \in \toset{S}$ via $s_1 \sim \Unif{S}$ and $s_i \sim \Unif{S \setminus \{s_1,\ldots,s_{i-1}\}}$.\footnote{For multisets $S$ and $T$, we write $S \supseteq T$ iff for every $x \in T$, $\mult{S}{x} \geq \mult{T}{x}$. In this case, $S \setminus T$ denotes the new multiset where for every $x \in S$, $\mult{S \setminus T}{x} = \begin{cases} \mult{S}{x} - \mult{T}{x} & x \in T \\ \mult{S}{x} & x \not\in T \end{cases}$. Thus, in particular, $|S \setminus T| = |S| - |T|$ where $|\cdot|$ denotes cardinality, i.e., the sum of multiplicities.} We contrast this with the standard ``with replacement'' product distribution $(\Unif{S})^t$. We have some more standard facts:

\begin{proposition}[``With replacement'' vs. ``without replacement'' sampling]\label{prop:prelim:with-vs-without-replacement}
    For every $t \in \BN$ and $\epsilon > 0$, there exists $m \in \BN$ such that for every multiset $S$ with $|S| \geq m$, \[
    \tvdist{\NoReplace{t}{S}}{(\Unif{S})^t} \leq \epsilon.
    \]
\end{proposition}
{Note that \Cref{prop:prelim:with-vs-without-replacement} is stated for multisets. However, the statement for multisets reduces immediately to the statement for sets: Given a multiset $S$ with $|S| = n$, consider an arbitrary ordering $\sigma \in \Orderings{S}$. Then observe that $\sigma^t \circ \NoReplace{t}{[n]} = \NoReplace{t}{S}$ and $\sigma^t \circ (\Unif{[n]})^t = (\Unif{S})^t$ (where $\sigma^t \circ \CD$ means to sample $(i_1,\ldots,i_t) \sim \CD$ and then output $(\sigma(i_1),\ldots,\sigma(i_t))$), and apply the data processing inequality.}

\begin{proposition}[Hoeffding's inequality]
\label{prop:hoeffding}
Let $X_1, \dots, X_n$ be independent random variables such that $a_i \leq X_i \leq b_i$ for all $i \in [n]$. For all $t > 0$, we have
\[
\Pr\paren*{ \abs*{ \sum_{i = 1}^n X_i - \sum_{i = 1}^n \Exp\bracket*{ X_i } } \geq t } \leq 2 \cdot \mathrm{e}^{ - \frac{ 2t^2 }{ \sum_{i = 1}^n \paren*{ b_i - a_i }^2 } } .
\]
\end{proposition}

\begin{proposition}[Fixing normalization]\label{prop:prelim:fix-norm}
    Let $S$ be a finite set and $\CD \in \Dist{S}$ a distribution. Let $Y : S \to \BR_{\geq 0}$ be a function such that
    \[
    \sum_{s \in S} \abs*{Y(s) - \CD(s)} \leq \epsilon.
    \]
    Let $\gamma = \sum_{s \in S} Y(s)$. Then $\hat{\CD}(s) = \frac{Y(s)}\gamma$ is a distribution with $\tvdist{\CD}{\hat{\CD}} \leq \epsilon$.
\end{proposition}

\begin{proof}
    First, by the triangle inequality $\abs*{\gamma - 1} = \abs*{\sum_{s \in S} Y(s) - 1} = \abs*{\sum_{s \in S} (Y(s) - \CD(s))} \leq \sum_{s \in S} \abs*{Y(s) - \CD(s)} \leq \epsilon.$ Further,
    $\sum_{s \in S} \abs*{\hat{\CD}(s) - Y(s)} = \sum_{s \in S} \abs*{\frac1{\gamma} Y(s) - Y(s)} = \sum_{s \in S} Y(s) \abs*{\frac1\gamma-1} = \gamma \abs*{\frac1\gamma-1} = \abs*{1-\gamma} \leq \epsilon.$ Hence finally, again using the triangle inequality:
    $\sum_{s \in S} \abs*{\hat{\CD}(s) - \CD(s)} \leq \sum_{s \in S} \paren*{ \abs*{\hat{\CD}(s) - Y(s)} + \abs*{Y(s) - \CD(s)}} \leq 2\epsilon$.
\end{proof}

\begin{proposition}
\label{claim:cls17exp-helper}
Let $\mathsf{X}$ be a random variable that takes values in the interval $[0, 1]$. Let $0 \leq \mu, \delta \leq 1$ be such that $\Pr\paren*{ \abs*{ \mathsf{X} - \mu } \geq \delta } \leq \delta$. Then:
\[
\abs*{ \Exp\bracket*{ \mathsf{X} } - \mu } \leq 2 \delta .
\]
\end{proposition}
\begin{proof}
Define an event $\mathcal{E}$ that occurs if and only if we have $\abs*{ \mathsf{X} - \mu } \geq \delta$. By our assumption, we have $\Pr\paren*{ \mathcal{E} } \leq \delta$. By the chain rule, we have:
\[
\Exp\bracket*{ \mathsf{X} } = \Pr\paren*{ \mathcal{E} } \cdot \Exp\bracket*{ \mathsf{X} \mid \mathcal{E} } + \Pr\paren*{ \overline{ \mathcal{E} } } \cdot \Exp\bracket*{ \mathsf{X} \mid \overline{ \mathcal{E} } } .
\]
By definition of $\mathcal{E}$, we have that, conditioned on $\overline{ \mathcal{E} }$, it holds that $\mathsf{X} \leq \mu + \delta$ with probability $1$. As $\mathsf{X}$ takes values in $[0, 1]$, $\Exp \bracket*{\mathsf{X} \mid \mathcal{E}} \leq 1$ and so we have:
\[
\Exp\bracket*{ \mathsf{X} } \leq \Pr\paren*{ \mathcal{E} } + \Exp\bracket*{ \mathsf{X} \mid \overline{\mathcal{E}}} \leq \mu + 2 \delta .
\]
Similarly, $\Exp \bracket*{\mathsf{X} \mid \mathcal{E}} \geq 0$ and so we have:
\[
\Exp\bracket*{ \mathsf{X} } \geq \Pr\paren*{ \overline{ \mathcal{E} } } \cdot \Exp\bracket*{ \mathsf{X} \mid \overline{ \mathcal{E} } } \geq \paren*{ 1 - \delta } \cdot \paren*{ \mu - \delta } \geq \mu - 2 \delta .
\]
\end{proof}

\section{Algorithms from neighborhood type sampling}\label{sec:nbr-types}

In this section, we develop some general techniques for reducing $(1/2-\epsilon)$-approximating the $\mdcut$ value of graphs to estimating certain neighborhood-type distributions in graphs.

\subsection{Preprocessing colors}

\begin{proposition}
\label{prop:coloring}
    Let $G = (V,E)$ be a graph and $k \geq 2 \in \BN$, and let $\delta > 0$. Let $\CK$ be any two-wise independent distribution over functions $\chi : V \to [k]$. (I.e., for all $u \neq v \in V$, and $i, j \in [k]$, $\Pr_{\chi \sim \CK}[\chi(u)=i \wedge \chi(v)=j] = 1/k^2$.) Then, \[ \Pr_{\chi \sim \CK}\left[\Pr_{e=(u,v) \sim \Unif{E}}[\chi(u) = \chi(v)] \geq \delta\right] \leq \frac1{\delta k}. \] In particular, the RHS is less than $1\%$ if $k\geq 100/\delta$.
\end{proposition}

\begin{proof}
    Enumerate $E$'s edges as $\{e_1,\ldots,e_m\}$, and let $X_j$ be the indicator for the event that $\chi(u_j) = \chi(v_j)$ (where $e_j = (u_j,v_j)$). Then by $2$-wise independence, $\Pr_\chi[X_j] = \sum_{c=1}^k \Pr_\chi[\chi(u_j) = \chi(v_j) = c] = k \cdot 1/k^2 = 1/k$. Then $\Pr_{e=(u,v) \sim \Unif{E}}[\chi(u) = \chi(v)] = \frac1m \sum_{j=1}^m X_j$, and therefore $\Exp_\chi[\Pr_{e=(u,v) \sim \Unif{E}}[\chi(u) = \chi(v)]] = 1/k$ by linearity of expectation. Finally, we apply Markov's inequality.
\end{proof}

We use this proposition to add a ``preprocessing'' step to all of our algorithms: Before we start running a streaming algorithm $\CA$, we sample a $2$-wise independent function $\chi : V \to [k]$; we give $\CA$ black-box access to $\chi$ (encoded as a polynomial's coefficients, which requires polylogarithmic bits); then immediately before each edge $(u,v)$ is processed by $\CA$, we discard the edge (i.e., do not pass it to $\CA$) if $\chi(u) = \chi(v)$. The input from $\CA$'s point of view is then a stream of graph edges together with black-box access to a proper coloring. Moreover, in the case of random-ordering algorithms, note that conditioned on $\chi$, we still provide a uniformly random ordering over the graph's remaining edges. In other words, the distributions ``sample a uniform random ordering of all edges, then sample a coloring and throw away the improperly-colored edges and output the remaining edges in order'' and ``sample a coloring and throw away the improperly-colored edges, then sample a uniformly random ordering of the remaining edges'' are identical.

\subsection{Induced subgraphs and neighborhoods}

We will require the following useful proposition about neighborhoods inside of induced subgraphs:

\begin{proposition}\label{prop:prelim:induced-nbhd}
    Let $k,\ell \in \BN$ and let $(G = (V,E), \chi : V \to[k])$ be a $k$-colored graph. Let $S \subseteq V$ with $v \in S$. The following are equivalent:
    \begin{enumerate}
        \item Every vertex $w \in \ball{G[S]}{\ell-1}{v}$ has $\deg{G[S]}{w} = \deg{G}{w}$.
        \item $\ball{G}{\ell}{v} \subseteq S$.
        \item $\ball{G[S]}{\ell}{v} = \ball{G}{\ell}{v}$.
    \end{enumerate}
\end{proposition}

\begin{proof}
    Observe that by the definition of induced subgraph, for every $w \in S$, $\deg{G[S]}{w} \leq \deg{G}{w}$. Further, $\ball{G[S]}{\ell}{v} \subseteq \ball{G}{\ell}{v}$.

    ($2 \implies 1$) Suppose there exists a vertex $w \in \ball{G[S]}{\ell-1}{v}$ with $\deg{G[S]}{w} < \deg{G}{w}$. Then there exists some vertex $z \in V$ incident to $w$ but not in $S$. But $w \in \ball{G}{\ell-1}{v}$ so $z \in \ball{G}{\ell}{v}$ (i.e., $w$ is within distance $\ell-1$ of $v$ so $z$ must be within distance $\ell$ of $v$).

    ($3 \implies 2$) By definition, $\ball{G[S]}{\ell}{v} \subseteq S$, hence by assumption $\ball{G}{\ell}{v} \subseteq S$.

    ($1 \implies 3$) Suppose there exists a vertex $z \in \ball{G}{\ell}{v} \setminus \ball{G[S]}{\ell}{v}$. Thus, there exists some path $v = w_0, \ldots, z = w_L$ in $G$ between $v$ and $z$ of length $L \leq \ell$ (each $w_i$ is incident to $w_{i-1}$). Note that $v \in \ball{G[S]}{\ell}{v}$ while $z \not\in \ball{G[S]}{\ell}{v}$. Thus, there exists some $i \in [L]$ such that $w_0,\ldots,w_{i-1} \in \ball{G[S]}{\ell}{v}$ but $w_i \not\in \ball{G[S]}{\ell}{v}$. By the former, $w_{i-1} \in \ball{G[S]}{i-1}{v}$, so we must have $w_i \not\in S$, else we would have $w_i \in \ball{G[S]}{i}{v}$, contradicting the latter. Since $w_{i-1} \in S$ but $w_i \not\in S$, we deduce that $\deg{G[S]}{w_{i-1}} < \deg{G}{w_{i-1}}$.
\end{proof}

\subsection{Edge-type distributions and the multi-pass algorithm}

Recall that the local $1/2$-approximation for $\mdcut$ of \textcite{BFSS15,CLS17} builds a random assignment to a $k$-colored graph $(G=(V,E),\chi : V\to[k])$'s vertices, where each vertex $v$ is assigned independently with a probability depending only on its local neighborhood. To design streaming $(1/2-\epsilon)$-approximation algorithms for $\mdcut$, we will be interested in using a streaming algorithm to simulate the \cite{BFSS15,CLS17} algorithm, or more precisely, to estimate the value of the cut it produces. This is equivalent to estimating the probability that a random edge in the graph is satisfied by the cut. But the probability that the algorithm's assignment satisfies a particular edge depends only on the neighborhood-types of its two endpoints. So we arrive at the goal of estimating the distribution of the endpoints' $(k+1)$-neighborhood types of a random edge.

\begin{definition}[Doubly-rooted colored graph]
    A \emph{doubly rooted (properly) $k$-colored graph} is a triple $(G=(V,E),\chi : V \to [k], e \in E)$, where $G$ is a graph, $\chi$ is a proper $k$-coloring of $G$, and $e$ is a designated \emph{root} edge.
\end{definition}

For a bijection $\phi : V \to V'$ and an edge $e = (u,v)$, we use $\phi(e)$ to denote the pair $(\phi(u),\phi(v))$.

\begin{definition}[Isomorphism of doubly rooted colored graphs]
    Two doubly rooted $k$-colored graphs $(G= (V,E),\chi,e)$ and $(G'=(V',E'),\chi', e')$ are \emph{isomorphic} if there exists a bijection $\phi : V \to V'$ such that (i) for all $u \neq v \in V$, $\mult{E}{u,v} = \mult{E'}{\phi(u,v)}$, (ii) for all $v \in V$, $\chi(v) = \chi'(\phi(v))$, and (iii) $\phi(e) = e'$.
\end{definition}

Note that, importantly, this notion of isomorphism does \emph{not} allow exchanging colors.\footnote{For instance, an isolated vertex colored $1$ is \emph{not} isomorphic to an isolated vertex colored $2$.} Also, isomorphism preserves the degrees of vertices, the distances between pairs of vertices, and the $\mdcut$ values of assignments.

Given a doubly rooted $k$-colored graph $(G,\chi,e)$, let $\type{G}{\chi}{e}$ denote the corresponding isomorphism class (``type'') of doubly rooted $k$-colored graphs. For $k,\ell,D \in \BN$, we let $\DbAllTypesDeg{k}{\ell}{D}$ denote the set of all isomorphism classes of $D$-bounded doubly rooted $k$-colored graphs where every vertex is distance $\leq \ell$ from (at least) one of the roots. $\DbAllTypesDeg{k}{\ell}{D}$ is a finite set by \cref{prop:prelim:ball-size}; we let $\NDbTypes{k}{\ell}{D} := |\DbAllTypesDeg{k}{\ell}{D}|$ denote its size.

Given $G = (V,E)$ and an edge $e=(u,v) \in E$, we let $\ball{G}{\ell}{e} := \ball{G}{\ell}{u}\cup\ball{G}{\ell}{v}$.

\begin{definition}[radius-$\ell$ neighborhood type of edge]
    Let $k,\ell \in \BN$, $G = (V,E)$, $\chi : V \to [k]$, and $e \in E$. The \emph{radius-$\ell$ neighborhood type} of $e$, denoted $\nbrtype{G}{\chi}{\ell}{e}$, is $\nbrtype{G}{\chi}{\ell}{e} := \type{G[N]}{\chi|_N}{e}$ where $N:=\ball{G}{\ell}{e}$.
\end{definition}

If $G$ is $D$-bounded then $\nbrtype{G}{\chi}{\ell}{e} \in \DbAllTypesDeg{k}{\ell}{D}$.

We also require a variant of \cref{prop:prelim:induced-nbhd} for doubly-rooted graphs:

\begin{proposition}\label{prop:prelim:induced-db-nbhd}
    Let $k,\ell \in \BN$ and let $(G = (V,E), \chi : V \to[k])$ be a $k$-colored graph. Let $S \subseteq V$ with $e = (u,v) \in E$ and $u,v \in S$. The following are equivalent:
    \begin{enumerate}
        \item Every vertex $w \in \ball{G[S]}{\ell-1}{e}$ has $\deg{G[S]}{w} = \deg{G}{w}$.
        \item $\ball{G}{\ell}{e} \subseteq S$.
        \item $\ball{G[S]}{\ell}{u} = \ball{G}{\ell}{v}$ and $\ball{G[S]}{\ell}{v} = \ball{G}{\ell}{v}$.
    \end{enumerate}
    Further, if any of these equivalent conditions holds, then $\nbrtype{G[S]}{\chi|_S}{\ell}{u,v} = \nbrtype{G}{\chi}{\ell}{u,v}$.
\end{proposition}

\begin{proof}
    The equivalence of the first three conditions follows immediately from \cref{prop:prelim:induced-nbhd}: Each condition is the conjunction of the two corresponding conditions from \cref{prop:prelim:induced-nbhd} for $u$ and $v$ (e.g., $\ball{G}{\ell}{u,v} \subseteq S$ iff $\ball{G}{\ell}{u} \subseteq S$ and $\ball{G}{\ell}{v} \subseteq S$). For the final implication, we reproduce the proof from \cref{prop:prelim:induced-nbhd}: By definition, $\nbrtype{G[S]}{\chi|_S}{\ell}{u,v} = \type{(G[S])[N]}{\chi|_S|_N}{u,v}$ where $N = \ball{G[S]}{\ell}{u,v}$, while $\nbrtype{G}{\chi}{\ell}{u,v} = \type{G[M]}{\chi|_M}{u,v}$ where $M = \ball{G}{\ell}{u,v}$. By assumption, $N=M$, so $G[M] = G[N] = (G[S])[N]$ by \cref{prop:prelim:doubly-induced-subgph} (and trivially $\chi|_S|_N = \chi|_N$).
\end{proof}

\begin{proposition}\label{prop:prelim:edge-ball-indep}
    For all $D, \ell \in \BN$, there exists $\Delta \in \BN$ with the following property. For every graph $G = (V,E)$ with maximum-degree $D$ and every $e \in E$,
    \[ |\{e' \in E : \ball{G}{\ell}{e} \cap \ball{G}{\ell}{e'} \neq \emptyset\}| \leq \Delta. \]
\end{proposition}

\begin{proof}
    If $\ball{G}{\ell}{e} \cap \ball{G}{\ell}{e'} \neq \emptyset$, there must be $w \in \verts{e}$ and $w' \in \verts{e'}$ such that $w' \in \ball{G}{2\ell}{w}$; hence, $v' \in \ball{G}{2\ell+2} v$ where $v,v'$ are arbitrary vertices in $\verts{e}$ and $\verts{e'}$, respectively. By \cref{prop:prelim:ball-size}, $|\ball{G}{2\ell+2}{v}| \leq 2D^{2\ell+3}$. Further, for any such $v'$, there can be at most $D$ neighbors in $E$ (by the maximum-degree assumption). This gives the bound for $\Delta := 2D^{2\ell+4}$.
\end{proof}

\subsection{Edge-type distribution}

Now, we define a notion of the \emph{edge-type distribution} in a graph, and describe how estimating this distribution suffices for approximating the $\mdcut$ value of a graph.

\begin{definition}[Edge-type distribution]
    Let $k,\ell,D \in \BN$ and $(G = (V,E, \chi : V \to [k])$ be a $D$-bounded $k$-colored graph. The \emph{radius-$\ell$ neighborhood type distribution} of $(G,\chi)$, denoted $\EdgeDist{G}{\chi}{\ell}$, is the distribution over $\DbAllTypesDeg{k}{\ell}{D}$ given by sampling a random $e \sim \Unif{E}$ and outputting $\nbrtype{G}{\chi}{\ell}{e}$.
\end{definition}

The fact that the edge-type distribution suffices for approximating the $\mdcut$ value of a graph is captured by the following theorem:

\begin{theorem}[Implied by \cite{BFSS15,CLS17}]\label{thm:types-to-dicut}
    Let $k,D \in \BN$. There exists a function $\Local : \DbAllTypesDeg{k}{k}{D} \to [0,1]$ such that the following holds. Let $(G=(V,E),\chi : V \to [k])$ be a $D$-bounded $k$-colored graph. Then \[ \frac12 \maxval{G} \leq \Exp_{T \sim \EdgeDist{G}{\chi}{k}}[\Local(T)] \leq \maxval{G}. \]
\end{theorem}

We remark that this theorem is not stated explicitly in the papers \cite{BFSS15,CLS17}, but is directly implied by these works. The key fact is that in the deterministic algorithm for producing a fractional cut presented in \cite[\S4]{BFSS15}, the assignment to each vertex depends only on the isomorphism class of the radius-$k$ neighborhood of that vertex; thus, the probability any edge is satisfied depends only on the isomorphism classes of the neighborhoods of its two endpoints, and the type of the edge is only more informative (it contains additional information about the intersection of these two neighborhoods). A variant of this algorithm, which also suffices to prove the above theorem, is \cref{algo:cls17} in \cref{sec:multipass} below with $\alpha=0$. 

Immediately from properties of the total variation distance (in particular, \cref{prop:prelim:tv-diff}), we deduce:

\begin{corollary}\label{cor:wrapper-DICUT}
    Let $k,D \in \BN$ and let $\Local : \DbAllTypesDeg{k}{k}{D} \to [0,1]$ be the function in the previous theorem. For every $\epsilon > 0$ and $(G=(V,E),\chi : V \to [k])$ a $D$-bounded $k$-colored graph, if $\CD \in \Dist{\DbAllTypesDeg{k}{k}{D}}$ is such that $\tvdist{\CD}{\EdgeDist{G}{\chi}{k}} \leq \epsilon$, then \[ \frac12 \maxval{G} - 2\epsilon \leq \Exp_{T \sim \CD}[\Local(T)] - \epsilon \leq \maxval{G}. \]
\end{corollary}

Now, we arrive at the main statements on estimating the edge-type distribution which we develop in the following sections.

\begin{theorem}[Single-pass adversarial-order estimator]\label{thm:est:adv-order}
For all $k, \ell, D \in \BN$ and $\epsilon,\delta > 0$, there exists $c > 0$ such that the following holds. There exists an $O(n^{1-c})$-space streaming algorithm that, for every $D$-bounded $k$-colored graph $(G = ([n], E), \chi : [n] \to [k])$, given (black-box access to) $\chi$ and a single, adversarially-ordered pass over $G$'s edges, outputs $\CD\in\Dist{\DbAllTypesDeg{k}{\ell}{D}}$ satisfying $\tvdist{\CD}{\EdgeDist{G}{\chi}{\ell}} \leq \epsilon$ except w.p. $\delta$.
\end{theorem}

\begin{theorem}[Single-pass random-order estimator]\label{thm:est:rand-order}
    For all $k, \ell, D \in \BN$ and $\epsilon,\delta > 0$, there exists $C > 0$ such that the following holds. There exists a $C \log n$-space streaming algorithm that, for every $D$-bounded $k$-colored graph $(G = ([n], E), \chi : [n] \to [k])$, given (black-box access to) $\chi$ and a single, randomly-ordered pass over $G$'s edges, outputs $\CD\in\Dist{\DbAllTypesDeg{k}{\ell}{D}}$ satisfying $\tvdist{\CD}{\EdgeDist{G}{\chi}{\ell}} \leq \epsilon$ except w.p. $\delta$.
\end{theorem}

We are now ready to prove \cref{thm:single-pass:adversarial-order} and \cref{thm:single-pass:random-order}, assuming \cref{thm:est:adv-order} and \cref{thm:est:rand-order}.

\begin{proof}[Proof of \cref{thm:single-pass:adversarial-order}]
Let $h:[n]\rightarrow [100/\epsilon]$ be a $2$-wise independent hash function that describes a coloring $\chi: V \rightarrow [100/\epsilon]$. It follows from \cref{prop:coloring} that with probability at least $9/10$, the fraction of monochromatic edges is at most $\epsilon/2$. We ``delete'' these monochromatic edges from the stream and compute the $\mdcut$ value of the remaining graph $\tilde{G}$. With probability at least $9/10$, $|\maxval{\tilde{G}}-\maxval{G}|\le \epsilon/2$. For the reduced graph $\tilde{G}$, the hash function $h$ gives black-box access to a proper coloring of its vertices. Thus, we can apply the algorithm from \cref{thm:est:adv-order} to output a distribution $\CD\in\Dist{\DbAllTypesDeg{k}{\ell}{D}}$ satisfying $\tvdist{\CD}{\EdgeDist{\tilde{G}}{\chi}{\ell}} \leq \epsilon/4$ except w.p. $1/10$. It follows from \cref{cor:wrapper-DICUT} that \[ \frac12 \maxval{\tilde{G}} - \epsilon/2 \leq \Exp_{T \sim \CD}[\Local(T)] - \epsilon/4 \leq \maxval{\tilde{G}}. \] Applying the union bound, we conclude that there is a streaming algorithm which $(1/2-\epsilon)$-approximates the $\mdcut$ value of a $D$-bounded $n$-vertex graph in $O(n^{1-\Omega(1)})$ space using a single, adversarially-ordered pass over the list of edges, with probability at least $8/10$.
\end{proof}

The proof of \cref{thm:single-pass:random-order} is analogous.

\section{Adversarial-ordering, $o(n)$-space, single-pass algorithm}\label{sec:adv-ord}

In this section, we prove \cref{thm:est:adv-order}, using \cref{alg:adv-ord} that we describe below. The algorithm uses \cref{alg:adv-ord:fixed-type} as a sub-routine. In addition to the inputs given to \cref{thm:est:adv-order}, the latter algorithm takes as input a target type $T \in \DbAllTypesDeg{k}{\ell}{D}$, and an estimate $\hat{m} \in \BN$ for the number of edges in the graph, and outputs an estimate for the probability mass of $T$ in $\EdgeDist{G}{\chi}{\ell}$.

\begin{algorithm}[H]
    \caption{$\textsf{Bounded-Degree-Adversarial}_D(n, k, \chi, \ell, \epsilon, \delta, \vecsigma)$}
    \label{alg:adv-ord}
    \begin{algorithmic}[1]
        \Statex \textbf{Parameters:} Number of vertices $n \in \BN$, number of colors $k \in \BN$, coloring $\chi : [n] \to [k]$, maximum degree $D \in \BN$, radius $\ell \in \BN$, accuracy $\epsilon > 0$, and failure probability $\delta > 0$.
        \Statex \textbf{Input:}  A stream of edges $\vecsigma$ from $G$
        \State Maintain a global counter for the number of edges $m$
        \For{every integer $b$ from $0$ to $\lfloor \log (nD/2)\rfloor$}
        \For{every type $T\in  \DbAllTypesDeg{k}{\ell}{D}$}
         \State $Y_{b,T} \gets \textsf{Bounded-Degree-Adversarial-FixedType}_D(n, k, \chi, \ell, \epsilon/\NDbTypes{k}{\ell}{D}, \delta/\NDbTypes{k}{\ell}{D}, \vecsigma,T,2^b)$
        \EndFor    
     \EndFor
     \State $\gamma \gets \sum_{T \in \DbAllTypesDeg{k}{\ell}{D}} Y_{b,T}$
     \State $\CN \gets (\frac1\gamma Y_{b,T})_{T\in \DbAllTypesDeg{k}{\ell}{D}}$ where $b = \lfloor \log m \rfloor$
    \State Output $\CN$

    \alglinenoNew{algcommon}
    \alglinenoPush{algcommon}
    \end{algorithmic}
\end{algorithm}

Given a type $T \in \DbAllTypesDeg{k}{\ell}{D}$, and an arbitrary representative $(G,\chi,r,s)$ of $T$, let $a(T) := |\ball{G}{\ell}{r,s}|$ (one easily verifies that this does not depend on the choice of representative).

\begin{algorithm}[H]
    \caption{$\textsf{Bounded-Degree-Adversarial-FixedType}_D(n, k, \chi, \ell, \epsilon, \delta, \vecsigma,T,\hat{m})$}
    \label{alg:adv-ord:fixed-type}
    \begin{algorithmic}[1]

        \alglinenoPop{algcommon}
        
        \Statex \textbf{Parameters:} Number of vertices $n \in \BN$, number of colors $k \in \BN$, coloring $\chi : [n] \to [k]$, maximum degree $D \in \BN$, radius $\ell \in \BN$, accuracy $\epsilon > 0$, failure probability $\delta > 0$, target type $T \in \DbAllTypesDeg{k}{\ell}{D}$, and estimate $\hat{m} \in \BN$ for number of edges.
        \Statex \textbf{Input:} A stream of edges $\vecsigma$ from $G$.
        \Statex
        \Statex\textbf{Pre-processing:}

        \State Set $\Delta$ to be the constant given by \cref{prop:prelim:edge-ball-indep}.

        \State Set $K \gets ((\epsilon^2 \delta \hat{m})/(2D \Delta))^{1/a(T)}$.

        \If{$K \leq 1$}
        \State Store the entire input stream $\vecsigma$ and return the exact fraction of type-$T$ edges
        \EndIf

        \State Let $H : [n] \to [2^{\lfloor \log_2 K \rfloor}]$ be a $2a(T)$-wise independent hash function

        \State Initialize $S \gets \emptyset$ (set) and $F \gets \emptyset$ (multiset)

        \State Initialize $c \gets 1/(2\cdot a(T))$

        \Statex

        \Statex\textbf{Stream processing:}
        \For {edge $e$ in stream}\label{line:adv-ord:strm-proc}
        
            \For{$v \in \verts{e}$}
                \If{$H(v) = 1$ and $v \not\in S$}\label{line:adv-ord:hash-check}
                    \State $S \gets S \cup \{v\}$
                    \State $\degs[v] \gets 0$
                \EndIf

                \If{$v \in S$}
                \State $\degs[v] \gets \degs[v]+1$
                \EndIf
                
            \EndFor

            \If{$\verts{e} \subseteq S$}
                \State $F \gets F \cup \{e\}$
            \EndIf
            \If{$|S|\ge n^{1-c}$} \label{line:adv-ord:space-check}
            \State Terminate and output \texttt{fail}
            \EndIf
        \EndFor

        \Statex

        \Statex\textbf{Post-processing:}

        \State $X \gets 0$\label{line:adv-ord:X-init}

        \For{$u \neq v \in V$}\label{line:adv-ord:check-F}
            \If{$(u,v) \in F$}
                \If{every vertex $w \in \ball{(S,F)}{\ell-1}{u,v}$ has $\degs[w] = \deg{(S,F)}{w}$}\label{line:adv-ord:check-degs}
    
                    \If{$\nbrtype{(S,F)}{\chi}{\ell}{u,v} = T$}\label{line:adv-ord:check-T}
                        \State $X \gets X + \mult{F}{(u,v)}$\label{line:adv-ord:X-inc}
                    \EndIf
                
                \EndIf
            \EndIf
        \EndFor

        \State \textbf{return} $K^{a(T)} X/m$.

        \alglinenoPush{algcommon}
        
    \end{algorithmic}
\end{algorithm}

The key correctness lemma of this section is the following:

\begin{lemma}[Correctness and space bound for \Cref{alg:adv-ord:fixed-type}]\label{lem:adv-ord:fixed-type}
    Let $k,\ell,D \in \BN$ and $\epsilon,\delta > 0$ be constants. Let $n \in \BN$ and let $(G = (V,E), \chi : V \to[k])$ be a $k$-colored, $D$-bounded graph. Let $m = |E|$ and $\hat{m} \in \BN$. For every possible ordering $\vecsigma$ of $E$, \cref{alg:adv-ord:fixed-type} runs in $\tilde{O}(n^{1-c})$ space. Further, if $m/2 < \hat{m} \leq m$, w.p. $\geq 1-\delta$ the algorithm gives an output $Y$ satisfying
    \[
    \abs*{Y - \EdgeDist{G}{\ell}{\chi}(T)} \leq \epsilon.
    \]
\end{lemma}

\begin{proof}
    It follows from the condition on \cref{line:adv-ord:space-check} that the \cref{alg:adv-ord:fixed-type} always uses at most $\tilde{O}(n^{1-c})$ space (since the maximum degree is at most $D$, $|F|\le D |S|$).

    Now, we analyze the correctness of a \emph{hypothetical} version of \cref{alg:adv-ord:fixed-type} that does not perform the check in \cref{line:adv-ord:space-check}. We show that this hypothetical version satisfies both the correctness condition and $|S| \leq n^{1-c}$ w.p. $\delta$, which in turn implies the correctness result for the real algorithm.

    Note that $K > 1$ as long as $\hat{m}$ is at least a constant; further, when $K \leq 1$, the algorithm is guaranteed to be correct deterministically and uses only constant space. We assume hereafter that $K > 1$.

    Every non-isolated vertex passes the check in \cref{line:adv-ord:hash-check} with probability $p := 1/2^{\lfloor \log_2 K \rfloor}$. Thus, $2/K \geq p \geq 1/K$.

    The effect of the loop on \cref{line:adv-ord:strm-proc} is to guarantee the following at the end of the stream:
    \begin{enumerate}
        \item $S = \{v \in [n] : \deg{G}{v} > 0 \text{ and } H(v) = 1\}$.\label{item:adv-ord:S}
        \item $F = E[S]$, i.e., $F$ is the subset of edges in $E$ with both endpoints in $S$ (with multiplicity). Thus, in particular, the pair $(S,F)$ is the induced subgraph $G[S]$.\label{item:adv-ord:F}
        \item For every $v \in S$, $\degs[v] = \deg{G}{v}$.
    \end{enumerate}

    Thus, \cref{line:adv-ord:check-degs} checks whether every vertex $w \in \ball{G[S]}{\ell-1}{u,v}$ has $\deg{G[S]}{w} = \deg{G}{w}$. By \cref{prop:prelim:induced-db-nbhd}, this is equivalent to $\ball{G[S]}{\ell}{u,v} \subseteq S$. Further, \cref{line:adv-ord:check-T} checks if $\nbrtype{G[S]}{\chi|_S}{\ell}{u,v}=T$.

    Now, we analyze the random variable $X$, which controls the output of the algorithm. $X$ is set to $0$ on \cref{line:adv-ord:X-init} and increased on \cref{line:adv-ord:X-inc} for pairs $(u,v), u \neq v$ satisfying certain conditions. Let $X_{(u,v)}$ denote the amount by which $X$ is increased on \cref{line:adv-ord:X-inc} during the $(u,v)$-iteration, so that $X = \sum_{u\neq v\in V} X_{(u,v)}$. We claim:

    \begin{claim}
        For every $u \neq v \in V$, the following holds: If $(u,v) \in E$, $\ball{G}{\ell}{u,v} \subseteq S$, and $\nbrtype{G}{\chi}{\ell}{u,v} = T$, then $X_{(u,v)} = \mult{E}{(u,v)}$, and otherwise, $X_{(u,v)} = 0$.
    \end{claim}

    \begin{proof}
        Note that $X_{(u,v)} = \mult{F}{(u,v)}$ if and only if the checks on \cref{line:adv-ord:check-degs,line:adv-ord:check-F,line:adv-ord:check-T} all pass, and otherwise $X_{(u,v)} = 0$. To begin, since $F \subseteq E$, if $(u,v) \not\in E$, then $(u,v) \not\in F$, so the check on \cref{line:adv-ord:check-F} fails and $X_{(u,v)} = 0$. Otherwise, $(u,v) \in E$, and by definition of the induced subgraph, $\mult{E}{(u,v)} = \mult{F}{(u,v)}$. Next, \cref{line:adv-ord:check-degs} checks whether $\ball{G[S]}{\ell}{u,v} \subseteq S$. Finally, \cref{line:adv-ord:check-T} checks if $\nbrtype{G[S]}{\chi|_S}{\ell}{u,v}=T$. By \cref{prop:prelim:induced-db-nbhd}, $\ball{G[S]}{\ell}{u,v} \subseteq S$ implies that $\nbrtype{G[S]}{\chi|_S}{\ell}{u,v}=\nbrtype{G}{\chi|_S}{\ell}{u,v}$, so the check on \cref{line:adv-ord:check-T} passes if and only if $\nbrtype{G}{\chi|_S}{\ell}{u,v} = T$, as desired.
    \end{proof}

    Now let $\CT := \{e \in \toset{E} : \nbrtype{G}{\chi}{\ell}e = T\}$ denote the set of edges of type $T$. (Hence $|\CT| = \EdgeDist{G}{\ell}{\chi}(T) \cdot m$.) By the claim, $X_{(u,v)}$ is zero unless $(u,v) \in \CT$, in which case it equals $\mult{E}{e}$ iff $\ball{G}{\ell}{u,v} \subseteq S$. Hence we forget $X_e$ for $e \not\in \CT$ and write $X = \sum_{e \in \CT} X_e$. Further, recall that $S$ is a set which (a) contains each vertex with probability $p$ and (b) is $2a(T)$-wise independent (and in particular $a(T) = |\ball{G}{\ell}{u,v}|$). Thus $\Exp[X] = p^{a(T)} |\CT|$ and $(\Exp[X])^2 = p^{2a(T)} |\CT|^2$. Next, we compute
    \begin{align*}
    \Exp[X^2] &= \sum_{e, e' \in \CT} \Exp[X_e X_{e'}] \\
    &= \sum_{e,e' \in \CT : X_e,X_{e'}\text{ independent }} \Exp[X_e]\Exp[X_{e'}] + \sum_{e,e' \in \CT : X_e,X_{e'}\text{ dependent }} \Exp[X_e X_{e'}] \\
    &\leq (\Exp[X])^2 + \sum_{e,e' \in \CT : X_e,X_{e'}\text{ dependent }} \Exp[X_e X_{e'}]. \tag{$X_e,X_{e'}$ non-negative}
    \end{align*}
    Hence,
    \[
        \Var[X] = \Exp[X^2] - (\Exp[X]^2) \leq \sum_{e,e' \in \CT : X_e,X_{e'}\text{ dependent }} \Exp[X_e X_{e'}].
    \]

    Now when are $X_e$ and $X_{e'}$ dependent? Suppose $e = (u,v)$ and $e' = (u',v')$. Recall, $X_e$ is determined by whether $\ball{G}{\ell}{u,v} \subseteq S$ and similarly $X_{e'}$ by whether $\ball{G}{\ell}{u',v'} \subseteq S$. Recall that $|\ball{G}{\ell}{u,v}| = |\ball{G}{\ell}{u',v'}| = a(T)$ (since both $e$ and $e'$ are assumed to have type $T$) and therefore, by $2a(T)$-wise independence of $S$, $X_e$ and $X_{e'}$ are independent unless $\ball{G}{\ell}{u,v} \cap \ball{G}{\ell}{u',v'} \neq \emptyset$. Hence by \cref{prop:prelim:edge-ball-indep}, $X_e$ is dependent on $X_{e'}$ for at most $2D^{2\ell+4}$ distinct edges $e'$. Since $X_{e'} \leq D$, we deduce
    \[
        \Var[X] \leq D \sum_{e,e' \in \CT : X_e,X_{e'}\text{ dependent }} \Exp[X_e] \leq 2D^{2\ell+5} \sum_{e \in \CT}\Exp[ X_e] = \Delta D \Exp[X] \, ,
    \]
    where $\Delta$ is the constant from \cref{prop:prelim:edge-ball-indep}.
    
    Set $t := \frac{\epsilon \sqrt{p^{a(T)} m}}{\sqrt{\Delta D}}$, so that \[ t \sqrt{\Var[X]} \leq \frac{\epsilon \sqrt{p^{a(T)} m}}{\sqrt{\Delta D}} \cdot \sqrt{\Delta D \Exp[X]} = \epsilon p^{a(T)} \sqrt{m |\CT|} \leq \epsilon p^{a(T)} m. \] Hence by Chebyshev's inequality, \[ \Pr[|X-\Exp[X]| \geq \epsilon p^{a(T)} m] \leq \frac1{t^2}. \] Now we expand \[ \frac1{t^2} = \frac{\Delta D}{\epsilon^2 p^{a(T)} m} = \frac{\Delta D}{\epsilon^2 (2\Delta D/(\epsilon^2 \delta \hat{m}) m} = \frac{\delta \hat{m}}{2m} \leq \delta/2 \] since $\hat{m} \leq m$.
    
    Finally, let us analyze the number of vertices in $S$. Note that $\Exp[|S|] = pn_+$, where $n_+$ denotes the number of non-isolated vertices in $G$. Since $p = O(1/m^{1/a(T)})$ by assumption, $p n_+ = O(n_+/m^{1/a(T)}) \leq O(m^{1-1/a(T)}) = O(n^{1-1/a(T)})$ using $n_+ \leq 2m$ and $m \leq Dn$. By Markov's inequality, $\Pr[|S|\ge 2\Exp[|S|]/\delta]\le \delta/2$. For sufficiently large $n$, $2\Exp[|S|]/\delta$ is at most $n^{1-c}$. 
    
    Applying the union bound, we conclude that \cref{alg:adv-ord:fixed-type} succeeds with probability at least $1-\delta$.
\end{proof}

\begin{proof}[Proof of \cref{thm:est:adv-order}]
Again, $D,k,\ell,\epsilon, \delta$ are constants independent of $n$. Since \cref{alg:adv-ord} makes at most $O(\log n)$ calls to \cref{alg:adv-ord:fixed-type}, by \Cref{lem:adv-ord:fixed-type}, the space usage of \cref{alg:adv-ord} is at most $\tilde{O}(n^{1-c})$. We now prove correctness.

By definition, \cref{alg:adv-ord} outputs $\CN\in \Dist{\DbAllTypesDeg{k}{\ell}{D}}$ where for $T\in \DbAllTypesDeg{k}{\ell}{D}$, 
\[\CN(T) = \frac{Y_{b,T}}{\gamma},\] for $b = \lfloor \log m \rfloor$ where $\gamma = \sum_{T \in \DbAllTypesDeg{k}{\ell}{D}} Y_{b,T}$. Note that for this value of $b$, $m/2<2^b\le m$. Since $Y_{b,T}$ is the output of running \Cref{alg:adv-ord:fixed-type} with parameters $(n, k, \chi, \ell, \epsilon/\NDbTypes{k}{\ell}{D}, \delta/\NDbTypes{k}{\ell}{D}, \vecsigma,T,2^b)$, by \Cref{lem:adv-ord:fixed-type}, for every $T\in \DbAllTypesDeg{k}{\ell}{D}$, with failure probability at most  $\delta/\NDbTypes{k}{\ell}{D}$, we have
\[
\abs{Y_{b,T} - \EdgeDist{G}{\ell}{\chi}(T)} \leq \epsilon/\NDbTypes{k}{\ell}{D}.
\]
Thus, by the union bound, with failure probability at most $\delta$,
\[
\sum_{T \in \DbAllTypesDeg{k}{\ell}{D}} \abs{Y_{b,T} - \EdgeDist{G}{\ell}{\chi}(T)} \leq \epsilon.
\]
Finally, we apply \Cref{prop:prelim:fix-norm}.
\end{proof}

\section{Random-ordering, $O(\log n)$-space, single-pass algorithm}\label{sec:rand-ord}

\newcommand{\visnbhd}[4]{\mathsf{visnbhd}_{#1;#2}^{#3}(#4)}
\newcommand{\vistype}[5]{\mathsf{visnbhdtype}_{#1,#2;#3}^{#4}(#5)}
\newcommand{\VisDist}[1]{\mathsf{VisTypeDist}(#1)}
\newcommand{\EdgeVisDist}[3]{\mathsf{EdgeNbhdVisTypeDist}_{#1,#2}^{#3}}
\newcommand{\TNRG}{\BT_{\texttt{NoReplace-Global}}}
\newcommand{\TNRL}{\BT_{\texttt{NoReplace-Local}}}
\newcommand{\TRL}{\BT_{\texttt{Replace-Local}}}

Given an edge $e = (u,v)$, we use $\verts{e} := \{u,v\}$ to denote the set of $e$'s endpoints.

In this section, we prove \cref{thm:est:rand-order}. Toward this, we develop a notion capturing the subset of an edge's neighborhood resulting from greedily building a component (out to a fixed radius) from edges in an ordered stream.

\subsection{Visible types}

For a multiset $E$ of edges and an edge $e \in E$, $E \setminus \{e\}$ denotes $E$ with only a single copy of the edge $e$ removed.

\begin{definition}[Visible neighborhood]\label{def:rand-ord:vis}
Let $(G=(V,E), \chi : V \to [k], e = (u,v))$ be a doubly-rooted graph, and let $\sigma \in \Orderings{E}$. For $\ell \geq 1 \in \BN$, consider the following process:
\begin{enumerate}
    \item Initialize $\overline{V} \gets \{u,v\}$, $\overline{E} \gets \emptyset$. ($\overline{E}$ is a multiset.)
    \item For $i = 1,\ldots,m$:
    \begin{enumerate}
        \item If $\verts{\sigma(i)} \cap \ball{(\overline{V},\overline{E})}{\ell-1}{e} \neq \emptyset$, add $\sigma(i)$ to $\overline{E}$ and $\verts{\sigma(i)}$ to $\overline{V}$.
    \end{enumerate}
\end{enumerate}
The \emph{radius-$\ell$ visible neighborhood} of $e$ given the ordering $\sigma$, denoted $\visnbhd{G}{\sigma}{\ell}{e}$, is the graph $(\overline{V}, \overline{E})$.
\end{definition}

Note that if $N = \ball{G}{\ell}{e}$, then $\overline{V} \subseteq N$ and $\overline{E} \subseteq E[N]$. Also, the multiplicity of $e$ and the degrees of its endpoints are all preserved in $\visnbhd{G}{\sigma}{\ell}{e}$.

\begin{example}
    Let $G = ([4],\{(1,2),(2,3),(3,4)\})$. If $\sigma = ((2,3),(3,4))$, then $\visnbhd{G}{\sigma}{2}{1,2}$ is $([4],\{(1,2),(2,3),(3,4)\})$. On the other hand, if $\sigma = ((1,2),(3,4),(2,3))$, then the graph of $\visnbhd{G}{\sigma}{2}{1,2}$ is $([3],\{(1,2),(2,3)\})$, since the edge $(3,4)$ is not ``visible'' from $(1,2)$ at the time it is encountered in the stream.
\end{example}

\begin{definition}[Visible type]
Let $(G=(V,E),\chi : V \to [k], e)$ be a doubly-rooted $k$-colored graph and $\sigma\in \Orderings{E}$. For $\ell \in \BN$, the \emph{radius-$\ell$ visible neighborhood type} of $e$ given the ordering $\sigma$, denoted $\vistype{G}{\chi}{\sigma}{\ell}{e}$, is $\type{(\overline{V},\overline{E})}{\chi_{\overline{V}}}{e}$, where $(\overline{V},\overline{E})=\visnbhd{G}{\sigma}{\ell}{e}$.
\end{definition}

As usual, if $G$ is $D$-bounded, then $\vistype{G}{\chi}{\sigma}{\ell}{e} \in \DbAllTypesDeg{k}{\ell}{D}$.

Let $\phi : E \to E'$ be a bijection between multisets $E$ and $E'$. (Properly speaking, $\phi$ is a bijection $\toset{E} \to \toset{E'}$ such that $\mult{E}{e} = \mult{E'}{\phi(e)}$ for all $e \in \toset{E}$.) We write $\phi(\sigma) \in \Orderings{E'}$ to denote the ordering $(\phi(\sigma))(i) = \phi(\sigma(i))$.

\begin{proposition}\label{prop:rand-ord:visnbhd-iso-inv}
    Let $k,\ell \in \BN$, let $(G=(V,E),\chi : V \to [k],e)$ and $(G'=(V',E'),\chi' : V' \to [k], e')$ be two isomorphic doubly-rooted $k$-colored graphs via a map $\phi : V \to V'$. Let $\sigma \in \Orderings{E}$. We have $\vistype{G}{\chi}{\sigma}{\ell}{e} = \vistype{G'}{\chi'}{\phi(\sigma)}{\ell}{e'}$.
\end{proposition}

(In this statement, we are using that $\phi$ induces a bijection $E \to E'$ in the natural way: $\phi(u,v) = (\phi(u),\phi(v))$.)

\begin{proof}
    Let $(\overline{V},\overline{E}) = \visnbhd{G}{\sigma}{\ell}{e}$ and $(\overline{V}',\overline{E}') = \visnbhd{G[N]}{\sigma[N]}{\ell}{e'}$. Imagine running the procedure in \cref{def:rand-ord:vis} to build up $(\overline{V},\overline{E})$ and $(\overline{V}',\overline{E}')$ in parallel. We claim that $\phi|_{\overline{V}}$ always yields an isomorphism between these two graphs. Initially, $\overline{V} = \verts{e},\overline{E}=\{\},\overline{V}'= \verts{e'},\overline{E}'=\{\}$, so $\phi|_{\overline{V}}$ indeed gives an isomorphism. We claim that after every step of the iteration in \cref{def:rand-ord:vis}, $\phi|_{\overline{V}}$ still gives an isomorphism. Indeed, at the $i$-th iteration, inductively we have $(\overline{V}',\overline{E}') = \phi((\overline{V},\overline{E}))$. Thus, the conditions ``$\verts{\sigma(i)} \cap \ball{(\overline{V},\overline{E})}{\ell-1}{e} \neq \emptyset$'' and ``$\verts{\phi(\sigma(i))} \cap \ball{\phi(\overline{V},\overline{E})}{\ell-1}{\phi(u),\phi(v)} \neq \emptyset$'' are equivalent, and so the edge is added in one graph iff it is added in the other. Finally, we observe that the operations ``add $\sigma(i)$ to $\overline{E}$ and $\verts{\sigma(i)}$ to $\overline{V}$'' and ``add $\phi(\sigma(i))$ to $\overline{E}'$ and $\verts{\phi(\sigma(i))}$ to $\overline{V}'$'' maintain the invariant that $\phi|_{\overline{V}}$ yields an isomorphism.
\end{proof}

Now, suppose $G = (V,E)$ is a graph and $\sigma \in \Orderings{E}$. For $U \subseteq V$, we write $\sigma[U] \in \Orderings{E[U]}$ for the \emph{induced ordering} on $U$, that is, the subordering consisting only of edges whose endpoints are both in $U$.

\begin{proposition}\label{prop:rand-ord:vis-nbhd}
    Let $(G = (V,E), \chi : V \to [k], e)$ be a doubly-rooted $k$-colored graph and $\sigma\in \Orderings{E}$. Then \[ \visnbhd{G}{\sigma}{\ell}{e} = \visnbhd{G[N]}{\sigma[N]}{\ell}{e} \] where $N := \ball{G}{\ell}{e}$.
\end{proposition}

\begin{proof}
    Let $(\overline{V},\overline{E}) = \visnbhd{G}{\sigma}{\ell}{e}$ and consider running the iterative procedure for calculating the visible neighborhood as in \cref{def:rand-ord:vis}. Consider also the calculation for $(\overline{V}',\overline{E}') = \visnbhd{G[N]}{\sigma[N]}{\ell}{e}$, which by \cref{def:rand-ord:vis} and the definition of induced orderings is equivalent to the following procedure:

    \begin{enumerate}
    \item Initialize $\overline{V}' \gets \verts{e}, \overline{E}' \gets \{\}$.
    \item For $i = 1,\ldots,m-1$:
    \begin{enumerate}
        \item If $\verts{\sigma(i)} \not\subseteq N$, skip this iteration.
        \item If $\verts{\sigma(i)} \cap \ball{(\overline{V}',\overline{E}')}{\ell-1}{e} \neq \emptyset$, add $\sigma(i)$ to $\overline{E}'$ and $\verts{\sigma(i)}$ to $\overline{V}'$.
    \end{enumerate}
    \end{enumerate}

    Note that initially, $\overline{V} = \overline{V}'$ and $\overline{E} = \overline{E}'$. Now we claim that after each loop iteration, these equalities still hold. Inductively, it suffices to prove that at every iteration, we ``add $\sigma(i)$ to $\overline{E}$ and $\verts{\sigma(i)}$ to $\overline{V}$'' iff we ``add $\sigma(i)$ to $\overline{E}'$ and $\verts{\sigma(i)}$ to $\overline{V}'$''. In other words, the condition ``$\verts{\sigma(i)} \cap \ball{(\overline{V},\overline{E})}{\ell-1}{e} \neq \emptyset$'' must be equivalent to the condition ``$\verts{\sigma(i)} \subseteq N$ and $\verts{\sigma(i)} \cap \ball{(\overline{V}',\overline{E}')}{\ell-1}{e} \neq \emptyset$''. But note that inductively, $\overline{V}'=\overline{V}$ and $\overline{E}' = \overline{E}$, and furthermore, if $\verts{\sigma(i)} \cap \ball{(\overline{V}',\overline{E}')}{\ell-1}{e} \neq \emptyset$, then $\verts{\sigma(i)} \subseteq \ball{(\overline{V}',\overline{E}')}{\ell}{e} \subseteq N $.
\end{proof}

\subsection{Distributions of visible types given random orderings}

Next, we consider what happens with a uniformly random ordering:

\begin{definition}[Visible type distribution of a type]
    Let $k, \ell,D \in \BN$, and let $T \in \DbAllTypesDeg{k}{\ell}{D}$ be a type. The \emph{visible type distribution} of $T$, denoted $\VisDist{T}$, is the following distribution on $\DbAllTypesDeg{k}{\ell}{D}$:
    \begin{enumerate}
        \item Pick any representative $(G = (V,E),\chi : V \to [k],e)$ for $T$.
        \item Sample $\sigma \sim \Unif{\Orderings{E}}$.
        \item Output $\vistype{G}{\chi}{\sigma}{\ell}{e}$.
    \end{enumerate}
\end{definition}

To see that this distribution is well-defined, i.e., does not depend on the choice of representative of $T$, observe that \cref{prop:rand-ord:visnbhd-iso-inv} says that if $(G = (V,E),\chi : V \to [k],e)$ and $(G' = (V',E'),\chi : V' \to [k],e')$ are isomorphic via $\phi : V \to V'$ then $\vistype{G}{\chi}{\sigma}{\ell}{e} = \vistype{G'}{\chi'}{\phi(\sigma)}{\ell}{e'}$. Further, if $\sigma \sim \Unif{\Orderings{E}}$, then $\phi(\sigma) \sim \Unif{\Orderings{E'}}$.

Our next proposition states that given a partition of a graph's vertices into blocks, sampling a single global ordering of all edges and looking at the induced ordering on each block is the same as sampling independent orderings on each block.

\begin{proposition}[Local vs. global sampling: orderings]\label{prop:rand-ord:induced-random}
Let $t \in \BN$ and $G=(V,E)$ be a graph. Let $e_1,\ldots,e_t \in E$ such that for $N_i := \ball{\ell-1}{G}{e_i}$, $N_i \cap N_j = \emptyset$ for $i \neq j$. The following distributions on $\prod_{i=1}^t \Orderings{E[N_i]}$ are identical:
\begin{enumerate}
    \item Sample $\sigma \sim \Unif{\Orderings{E}}$ and output $(\sigma[N_1],\ldots,\sigma[N_t])$.
    \item Sample $(\sigma^1,\ldots,\sigma^t)$ with $\sigma^i \sim \Unif{\Orderings{E[N_i]}}$ for $i \in [t]$ independently.
\end{enumerate}
\end{proposition}

We now extend this statement to hold for visible types in a graph: If we start from a ``spread-out'' set of base edges $e_1,\ldots,e_t$ (i.e., such that the radius-$\ell$ balls around each pair are disjoint), sampling a single global ordering on $E$ and looking at the visible type from each base edge is the same as sampling individual orderings on the neighborhood of each base edge:

\begin{proposition}[Local vs. global sampling: visible types]\label{prop:rand-ord:strm-random}
    Let $k,\ell,D \in \BN$ and let $(G=(V,E),\chi:V\to[k])$ be a $D$-bounded $k$-colored graph. Let $\{e_1,\ldots,e_t\} \subseteq E$ be such that for each $i \neq j \in [t]$, $\ball{G}{\ell}{e_i} \cap \ball{G}{\ell}{e_j} = \emptyset$. Let $\tau = (e_1,\ldots,e_t)$. Then the following distributions on $(\DbAllTypesDeg{k}{\ell}{D})^t$ are identical:
    \begin{enumerate}
        \item Sample $\sigma' \sim \Unif{\Orderings{E \setminus \{e_1,\ldots,e_t\}}}$. Then output \[ (\vistype{G}{\chi}{\tau \| \sigma'}{\ell}{e_1},\ldots,\vistype{G}{\chi}{\tau \| \sigma'}{\ell}{e_t}), \]
        where $\tau \| \sigma'$ denotes the ordering on $E$ given by concatenation.
        \item Output a sample from the product distribution \[ \VisDist{\nbrtype{G}{\chi}{\ell}{e_1}} \times \cdots \times \VisDist{\nbrtype{G}{\chi}{\ell}{e_t}}. \]
    \end{enumerate}
\end{proposition}

\begin{proof}
    Let $N_i := \ball{G}{\ell}{e_i}$ denote the $\ell$-neighborhood of the $i$-th edge. Note that $(E \setminus \{e_1,\ldots,e_t\})[N_i] = (E[N_i]) \setminus \{e_i\}$ by disjointness.
    
    Now consider the first listed distribution. By \cref{prop:rand-ord:vis-nbhd}, for each $i \in [t]$, \[ \vistype{G}{\chi}{\tau \| \sigma'}{\ell}{e_i} = \vistype{G[N_i]}{\chi|_{N_i}}{(\tau \| \sigma')[N_i]}{\ell}{e_i}. \]

    Remember $\sigma'$ is an ordering on $E \setminus \{e_1,\ldots,e_t\}$. $\tau \| \sigma'$ is an ordering on $E$. In both cases, we can take an induced ordering on $N_i$ to get an ordering on $E[N_i]$. Now we observe that the induced ordering $(\tau \| \sigma')[N_i]$ on $G[N_i]$ is the same as $(\tau[N_i]) \| (\sigma'[N_i]) = e_i \| (\sigma'[N_i])$. Hence \[
    \vistype{G[N_i]}{\chi|_{N_i}}{(\tau \| \sigma')[N_i]}{\ell}{e_i} = \vistype{G[N_i]}{\chi|_{N_i}}{e_i \| \sigma'[N_i]}{\ell}{e_i}.
    \]
    Now recall that $\sigma' \sim \Unif{\Orderings{E \setminus \{e_1,\ldots,e_t\}}}$. Letting $\sigma^i = \sigma'[N_i]$, by \cref{prop:rand-ord:induced-random}, $(\sigma^1,\ldots,\sigma^t)$ is also distributed as $\sigma^i$ being drawn from $\Unif{\Orderings{E[N_i]\setminus\{e_i\}}}$ independently. Hence the first distribution is equivalent to sampling uniformly random and independent orderings $(\sigma^1,\ldots,\sigma^t)$ on $E[N_1] \setminus \{e_1\},\ldots,E[N_t] \setminus \{e_t\}$, respectively, and outputting \[ (\vistype{G[N_1]}{\chi|_{N_1}}{e_1 \| \sigma^1}{\ell}{e_1},\ldots,\vistype{G[N_t]}{\chi|_{N_ti}}{e_t \| \sigma^t}{\ell}{e_t}). \] Since visible neighborhoods are invariant to the position of the base edge in the stream, this is precisely the product distribution $(T_1,\ldots,T_t)$ where $T_i \sim \VisDist{\nbrtype{G}{\chi}{\ell}{e_i}}$ independently, as desired.
\end{proof}

\subsection{The random-ordering algorithm}

\begin{definition}[Visible type distribution of a graph]
    Let $k, \ell, D \in \BN$ and let $(G=(V,E), \chi : V \to [k])$ be a $k$-colored graph with maximum-degree $D$. The \emph{visible edge neighborhood-type distribution} of $G$, denoted $\EdgeVisDist{G}{\chi}{\ell}$, is the distribution over $\DbAllTypesDeg{k}{\ell}D$ given by sampling $T \sim \EdgeDist{G}{\chi}{\ell}$ and then outputting a sample from $\VisDist{T}$.
\end{definition}

\begin{algorithm}[H]
    \caption{Random-ordering estimator of visible type distribution}\label{alg:rand-ord:vis}
    \begin{algorithmic}[1]
        \alglinenoPop{algcommon}
        
        \Statex \textbf{Parameters:} Number of vertices $n \in \BN$, number of colors $k \in \BN$, coloring $\chi : [n] \to [k]$, maximum degree $D \in \BN$, radius $\ell \in \BN$, accuracy $\epsilon > 0$, failure probability $\delta > 0$, and sampling parameter $t \in \BN$.
        \Statex \textbf{Input:} Randomly-ordered stream of edges $\{(e_i)\}$ from $G$.
        
        \Statex
        \Statex\textbf{Stream processing, phase 1:}
        \For {edge $e$ in first $t$ edges of stream}
            \State $e_i \gets e$
            \State $\overline{V}_i \gets \verts{e_i}$ (set)
            \State $\overline{E}_i \gets \{\}$ (multiset)
        \EndFor

        \Statex
        \Statex\textbf{Stream processing, phase 2:}
        \For {edge $e =e_1,\ldots,e_t$ and then remaining edges in stream}
            \For {$i = 1,\ldots,t$}
                \If {$\verts{e} \cap \ball{(\overline{V}_i,\overline{E}_i)}{\ell-1}{e_i} \neq \emptyset$} 
                    \State $\overline{E}_i \gets \overline{E}_i \cup \{e\}$
                    \State $\overline{V}_i \gets \overline{V}_i \cup \verts{e}$
                \EndIf
            \EndFor
        \EndFor

        \Statex
        \Statex\textbf{Post-processing:}

        \For {$i = 1,\ldots,t$}
        \State $T_i := \type{(\overline{V}_i,\overline{E}_i)}{\chi|_{\overline{V}_i}}{e_i}$\label{line:rand-ord:compute-types}
        \EndFor
        
        \State \textbf{return} $\EmpDist{\DbAllTypesDeg{k}{\ell}{D}}{T_1,\ldots,T_t}$.

        \alglinenoPush{algcommon}
        
    \end{algorithmic}
\end{algorithm}

\begin{theorem}\label{thm:rand-ord:vis-est}
    For all $k, \ell, D \in \BN$ and $\epsilon,\delta > 0$, there exists an $O(\log n)$-space streaming algorithm that, for every $k$-colored graph $(G = ([n], E), \chi : [n] \to [k])$ with maximum degree $D$, given $\chi$ and a single, randomly-ordered pass over $G$'s edges, outputs $\CD \in \Dist{\DbAllTypesDeg{k}{\ell}{D}}$ satisfying $\tvdist{\CD}{\EdgeVisDist{G}{\chi}{\ell}} \leq \epsilon$ except w.p. $\delta$.
\end{theorem}

\begin{proof}
    Let $t \in \BN$ be a constant to be determined later. Fix a graph $G = ([n], E)$ and a coloring $\chi : [n] \to [k]$. Consider running \cref{alg:rand-ord:vis} given the parameters $n,k,\chi,D,\ell,\epsilon,\delta,t$. We claim that there exists $t$ such that for sufficiently large $n$, the output $\hat{\CD}$ of the algorithm satisfies the desideratum except w.p. $\delta$. This probability is only over the choice of the random ordering of the stream of the edges $E$ --- \cref{alg:rand-ord:vis} is deterministic.

    Now, sampling uniform $\sigma \sim \Unif{\Orderings{E}}$ is equivalent to sampling $(e_1,\ldots,e_t) \sim \NoReplace{t}{E}$, sampling $\sigma' \sim \Unif{\Orderings{E \setminus \{e_1,\ldots,e_t\}}}$ independently, and then concatenating these orderings. Thus running \cref{alg:rand-ord:vis} is equivalent to the following:
    \begin{enumerate}
        \item The first phase samples $(e_1,\ldots,e_t) \sim \NoReplace{t}{E}$. Let $\tau^i = (e_1,\ldots,e_t)$.
        \item In the second phase, we sample $\sigma' \sim \Unif{\Orderings{E \setminus \{e_1,\ldots,e_t\}}}$. By the definition of visible types, and the disjointness of neighborhoods, for $i \in [t]$ we set $(\overline{V}_i,\overline{E}_i) = \visnbhd{G}{\tau \| \sigma'}{\ell}{e_i}$.
    \end{enumerate}
    Thus, the entire algorithm uses only $O(\log n)$ space, since each $(\overline{V}_i,\overline{E}_i)$ grows only to constant size. Furthermore, $T_i = \vistype{G}{\chi}{\sigma'}{\ell}{e_i}$ on \cref{line:rand-ord:compute-types}.

    Now, let $\TNRG$ denote the distribution of $\vecT = (T_1,\ldots,T_t)$ as defined on \cref{line:rand-ord:compute-types}. Rehashing the previous paragraph, $\TNRG$ is equivalent to the following two-step process:

    \begin{enumerate}
        \item Sample $(e_1,\ldots,e_t) \sim \NoReplace{t}{E}$. Let $\tau = (e_1,\ldots,e_t)$.
        \item Sample a uniform ordering $\sigma' \sim \Orderings{E \setminus \{e_1,\ldots,e_t\}}$, set $T_i \gets \vistype{G}{\chi}{\tau \| \sigma'}{\ell}{e_i}$, and output $(T_1,\ldots,T_t)$.
    \end{enumerate}
    The algorithm's correctness is then equivalent to the statement that
    \begin{equation}\label{eq:rand-ord:E}
        \Pr_{\vecT \sim \TNRG}[\tvdist{\EmpDist{\DbAllTypesDeg{k}{\ell}{D}}{\vecT}}{\EdgeVisDist{G}{\chi}{\ell}} \geq \epsilon] \leq \delta.
    \end{equation}
    Intuitively, the algorithm is correct because (i) sampling $(e_1,\ldots,e_t)$ without replacement is close to sampling $(e_1,\ldots,e_t)$ with replacement (\cref{prop:prelim:with-vs-without-replacement}) and (ii) \emph{conditioned} on the (highly likely) event that the radius-$\ell$ balls around $(e_1,\ldots,e_t)$ are all disjoint, the visible types of $(e_1,\ldots,e_t)$ given a random ordering $\sigma$ of the stream are random draws from $\VisDist{\nbrtype{G}{\chi}{\ell}{e_1}},\ldots,\VisDist{\nbrtype{G}{\chi}{\ell}{e_t}}$ (\Cref{prop:rand-ord:strm-random}). We formalize this intuition below.

    In particular, we introduce two ``hybrid'' modifications of $\TNRG$. Let $\TNRL$ denote the following distribution:
    \begin{enumerate}
        \item Sample $(e_1,\ldots,e_t) \sim \NoReplace{t}{E}$.
        \item Sample $T_i \sim \VisDist{\nbrtype{G}{\chi}{\ell}{e_i}}$ for $i \in [t]$ independently and output $(T_1,\ldots,T_t)$.
    \end{enumerate}
    Let $\TRL$ denote the distribution:
    \begin{enumerate}
        \item Sample $(e_1,\ldots,e_t) \sim (\Unif{E})^t$.
        \item Sample $T_i \sim \VisDist{\nbrtype{G}{\chi}{\ell}{e_i}}$ for $i \in [t]$ independently and output $(T_1,\ldots,T_t)$.
    \end{enumerate}
    Note that $\TNRG$, $\TNRL$, and $\TRL$ are all two-step sampling processes which first sample a $t$-tuple of edges and then sample corresponding ``visible types''. $\TNRG$ and $\TNRL$ differ only in the second step, while $\TNRL$ and $\TRL$ differ only in the first step.

    To prove the algorithm's correctness (\cref{eq:rand-ord:E}), we prove that there exists $t \in \BN$ such that:
    \begin{equation}\label{eq:rand-ord:E*-corr}
        \Pr_{\vecT \sim \TRL}[\tvdist{\EmpDist{\DbAllTypesDeg{k}{\ell}{D}}{\vecT}}{\EdgeVisDist{G}{\chi}{\ell}} \geq \epsilon] \leq \delta/3,
    \end{equation}
    and that for large enough $n \in \BN$,
    \begin{align}
        \tvdist{\TNRG}{\TNRL} &\leq \delta/3 \label{eq:rand-ord:E-E'} \\
        \tvdist{\TNRL}{\TRL} &\leq \delta/3 \label{eq:rand-ord:E'-E*}.
    \end{align}
    Together, these imply \cref{eq:rand-ord:E}. Indeed, define the ``bad event'' indicator $I : (\DbAllTypesDeg{k}{\ell}{D})^t \to \{0,1\}$ via $I(\vecT) := \1[\tvdist{\EmpDist{\DbAllTypesDeg{k}{\ell}{D}}{\vecT}}{\EdgeVisDist{G}{\chi}{\ell}} \geq \epsilon]$. Then we have
    \begin{multline*}
    \Exp_{\vecT \sim \TRL}[I(\vecT)] \leq \delta/3, \left|\Exp_{\vecT \sim \TNRG}[I(\vecT)] - \Exp_{\vecT \sim \TNRL}[I(\vecT)]\right| \leq \delta/3, \\
    \text{ and } \left|\Exp_{\vecT \sim \TNRL}[I(\vecT)] - \Exp_{\vecT \sim \TRL}[I(\vecT)]\right| \leq \delta/3,
    \end{multline*}
    via \cref{eq:rand-ord:E*-corr,prop:prelim:tv-diff,eq:rand-ord:E-E',eq:rand-ord:E'-E*}, and so the triangle inequality gives \[ \Exp_{\vecT \sim \TNRG}[I(\vecT)] = \Pr_{\vecT \sim \TNRG}[\tvdist{\EmpDist{\DbAllTypesDeg{k}{\ell}{D}}{\vecT}}{\EdgeVisDist{G}{\chi}{\ell}} \geq \epsilon] \leq \delta, \] as desired.

    Finally, it remains to prove these three inequalities. 
    
    \paragraph{Proving \cref{eq:rand-ord:E*-corr}.} We observe that in the definition of $\TRL$, each $T_i$ is distributed independently as $\EdgeVisDist{G}{\chi}{\ell}$, and for sufficiently large $t$, the empirical distribution will be $\epsilon$-close to the true distribution except with probability $\delta/3$ by \cref{prop:prelim:tv-samples}.

    \paragraph{Proving \cref{eq:rand-ord:E-E'}.} Note that if $e_1,\ldots,e_t$ have the property that for all $i \neq j \in [t]$, $\ball{G}{\ell}{e_i} \cap \ball{G}{\ell}{e_j} = \emptyset$, then by \cref{prop:rand-ord:strm-random}, the $T_i$'s are sampled independently as $\VisDist{\nbrtype{G}{\chi}{\ell}{e_i}}$, just as in $\TNRL$. So, by \cref{prop:prelim:sample-diff}, \[
    \tvdist{\TNRG}{\TNRL} \leq \Pr_{(e_1,\ldots,e_t)\sim\CE} [\exists i \neq j \in [t] : \ball{G}{\ell}{e_i} \cap \ball{G}{\ell}{e_j} \neq \emptyset]. \] But taking a union bound over $i \neq j$ and then conditioning on $e_i$, we can upper-bound this probability by $\binom{t}2 \frac{\Delta}{m-1}$, where $\Delta$ is the constant in \cref{prop:prelim:edge-ball-indep}. This bound is smaller than $\delta/3$ if $m \gg \Delta t^2/\delta$.
    
    \paragraph{Proving \cref{eq:rand-ord:E'-E*}.} Observe that $\TNRL$ and $\TRL$ are defined by sampling $(e_1,\ldots,e_t)$ from the ``without replacement'' distribution $\NoReplace{t}{E}$ and the ``with replacement'' $(\Unif{E})^t$, respectively, and then outputting the empirical distribution of $\{T_i \sim \VisDist{\nbrtype{G}{\chi}{\ell}{e_i}} : i \in [t]\}$. By \cref{prop:prelim:with-vs-without-replacement}, $\tvdist{\NoReplace{t}{E}}{(\Unif{E})^t} \leq \epsilon$. Hence by \cref{prop:prelim:data-proc}, $\tvdist{\TNRL}{\TRL} \leq \epsilon$.
\end{proof}

Finally, we arrive at:

\begin{proof}[Proof of \cref{thm:est:rand-order}]
    By invoking \cref{thm:rand-ord:vis-est}, we are given $\CD_{vis} \in \Dist{\DbAllTypesDeg{k}{\ell}{D}}$ such that $\tvdist{\CD_{vis}}{\EdgeVisDist{G}{\chi}{\ell}} \leq \epsilon$, and want to use this to build an estimate $\CD \in \Dist{\DbAllTypesDeg{k}{\ell}{D}}$ for $\EdgeDist{G}{\chi}{\ell}$. Our basic strategy is to view this as a \emph{linear algebraic} problem and apply some (extremely simple) tools from linear algebra.

    To begin, we expand by definition:
    \begin{align*}
    \EdgeVisDist{G}{\chi}{\ell}(T') &= \Exp_{T \sim \EdgeDist{G}{\chi}{\ell}} [(\VisDist{T})(T')] \\
    &= \sum_{T \in \DbAllTypesDeg{k}{\ell}{D}} \EdgeDist{G}{\chi}{\ell}(T) \cdot (\VisDist{T})(T').
    \end{align*}
    We can write this as a matrix-vector equation $\vecv_{vis} = M\vecv$, where $\vecv,\vecv_{vis} \in \BR^{\DbAllTypesDeg{k}{\ell}{D}}$ are vectors indexed by $\DbAllTypesDeg{k}{\ell}{D}$ and defined by $v_{vis}(T) = \EdgeVisDist{G}{\chi}{\ell}(T)$ and $v(T) = \EdgeDist{G}{\chi}{\ell}(T)$, respectively, and $M \in \BR^{\DbAllTypesDeg{k}{\ell}{D} \times \DbAllTypesDeg{k}{\ell}{D}}$ is indexed by pairs $\DbAllTypesDeg{k}{\ell}{D} \times \DbAllTypesDeg{k}{\ell}{D}$ given by $M(T',T) = (\VisDist{T})(T')$.
    
    Now observe that $M$ is invertible. To see this, as in \cite{MMPS17}, we define a partial ordering on $\DbAllTypesDeg{k}{\ell}{D}$ via $T' \preceq T$ iff any representative of $T'$ is isomorphic to any subgraph of any representative of $T$, and observe that $(\VisDist{T})(T') > 0$ iff $T' \preceq T$. Then, we pick any total ordering on $\DbAllTypesDeg{k}{\ell}{D}$ respecting the partial ordering; with respect to this ordering, $M$ is upper-triangular and has nonzero diagonal entries and such matrices are always diagonalizable.

    Let $\lambda$ denote the $1$-to-$1$ norm of $M^{-1}$, which is a constant depending only on $k,\ell,D$.\footnote{The $1$-to-$1$ norm of a matrix $M$ is $\sup_{\|x\|_1 = 1} \|Mx\|_1$.} Let $\vecd_{vis} \in \BR^{\DbAllTypesDeg{k}{\ell}{D}}$ be defined by $d_{vis}(T) = \CD_{vis}(T)$. We let $\vecd' \in \BR^{\DbAllTypesDeg{k}{\ell}{D}}$ be defined by $\vecd' := M^{-1} \vecd_{vis}$. By assumption on $\CD_{vis}$, $\|\vecd_{vis} - \vecv_{vis}\|_1 \leq 2\epsilon$ (cf. the definition of total variation distance, \cref{def:prelim:tv-dist}). 
    
    Finally, we bound $\|\vecd'-\vecv\|_1 = \|M^{-1} \vecd_{vis} - M^{-1} \vecv_{vis} \|_1 \leq \lambda \|\vecd_{vis} - \vecv_{vis}\|_1 \leq 2\lambda \epsilon$ since $\lambda$ is the $1$-to-$1$ norm of $M^{-1}$. Hence by \Cref{prop:prelim:fix-norm}, $\hat{\CD}(T) := d'(T)/(\sum_{T'} d'(T'))$ is a probability distribution satisfying $\tvdist{\hat{\CD}}{\Dist{\DbAllTypesDeg{k}{\ell}{D}}} \leq \|\vecd' - \vecv\|_1 \leq 4\lambda \epsilon$. The result follows from reparametrizing $\epsilon$.
\end{proof}


\section{Adversarial-ordering, $O(\log n)$-space, $O(1/\epsilon)$-pass algorithm}
\label{sec:multipass}


\newcommand\aPos{\mathtt{pos}}
\newcommand\aOpt{\mathtt{opt}}
\newcommand\aPiece{\mathtt{pre}}
\newcommand\aZero{\mathtt{zero}}
\newcommand\aOne{\mathtt{one}}
\newcommand\aX{\mathtt{x}}

\newcommand\aPosEst{\mathsf{P}}

\def\bits{\set*{ 0, 1 }}

\def\cutoff{\mathsf{cutoff}}
\def\cutest{\mathsf{Cut}\text{-}\mathsf{Est}}

\newcommand{\dirIn}{{\mathtt{in}}}
\newcommand{\dirOut}{{\mathtt{out}}}
\newcommand{\dirLow}{{\mathtt{lo}}}
\newcommand{\dirHigh}{{\mathtt{hi}}}

\newcommand{\vY}[3]{y_{#1} \paren*{#2,#3}} 
\newcommand{\vYNoAlpha}[2]{y_{#1} \paren*{#2}} 
\newcommand{\vZ}[2]{z_{#1}^\dirLow \paren*{#2}} 
\newcommand{\vA}[3]{a_{#1}\paren*{#2,#3}}

\newcommand{\vZEst}[2]{Z_{#1}^\dirLow \paren*{#2}} 

The goal of this section is to prove \cref{thm:multi-pass}, as formalized below:
\begin{theorem}[Formal version of \cref{thm:multi-pass}]
\label{thm:multipass-formal}
For all $\epsilon > 0$ small enough, there exists a randomized streaming algorithm that on input an $n$-vertex directed graph $G$ uses $\frac{200}{ \epsilon }$ passes and $\epsilon^{ - \frac{ 10^5 }{ \epsilon } } \cdot \log n$ space and outputs a value $\cutest$ satisfying:
\[
\Pr\paren*{ \paren*{ \frac{1}{2} - \epsilon } \cdot \maxval{G} \leq \cutest \leq \maxval{G} } \geq \frac{9}{10} .
\]
\end{theorem}
The proof of \cref{thm:multipass-formal} can be found in \cref{sec:multipass:formal}. As we build up to it, we first describe an algorithm inspired by \cite{BFSS15,CLS17} that additionally takes a parameter $k$ and a proper coloring $\chi$ of $G$ that has $k$ colors. This algorithm is for general $k$ although we will only use it for $k = O\paren*{ \frac{1}{ \epsilon } }$, where $\epsilon > 0$ is as in \cref{thm:multipass-formal} and is fixed for the rest of this section. We will also assume without loss of generality that $\epsilon$ is small enough, specifically that $\epsilon < 0.01$.

\subsection{A (non-streaming) algorithm for colored graphs}
\label{sec:multipass:cls17}

As mentioned above, we first describe an algorithm that heavily builds on \cite{BFSS15,CLS17}. Note that this algorithm is not in the streaming setting and that our description is different from the original. Its input is a directed multigraph $G = \paren*{ V, E }$ and a (proper) coloring $\chi : V \to [k]$ of the vertices in $G$ with $k$ colors (for some $k > 0$). In addition, it also takes as input a parameter $\alpha \geq 0$. From \cref{sec:multipass:cls17estimate} onwards, we will only work with $\alpha = \epsilon^5$, but for the sake of generality, this section is written for an arbitrary\footnote{We emphasize that $\alpha = 0$ is possible and essentially recovers the original \cite{BFSS15,CLS17} algorithm.} $\alpha \geq 0$.

The goal of this algorithm is to output a ``position'' $\aPos(v) \in [0, 1]$ for every vertex $v \in V$. The function $\aPos : V \to [0,1]$ can be viewed as a fractional cut in $G$, and \cite{BFSS15,CLS17} show that in the $\alpha = 0$ case, $\val{G}{\aPos}$ is at least $1/2$ of $\maxval{G}$, the maximum possible value of any cut. (We recover this result in \Cref{lemma:multipass:cls17} below.)

For the purposes of this section, we assume without loss of generality that the graph $G$ has no isolated vertices. We first partition the (multiset of) edges incident at a vertex $v \in V$ into four parts, depending on the directions and colorings, as follows:
\begin{equation}
\begin{aligned}
E_{ \dirIn, \dirLow }(v) &\coloneqq \set*{ (u, v) \in E \mid \chi(u) < \chi(v) }, &\hspace{1cm} E_{ \dirOut, \dirLow }(v) &\coloneqq \set*{ (v, u) \in E \mid \chi(u) < \chi(v) } , \\
E_{ \dirIn, \dirHigh }(v) &\coloneqq \set*{ (u, v) \in E \mid \chi(u) > \chi(v) }, &\hspace{1cm} E_{ \dirOut, \dirHigh }(v) &\coloneqq \set*{ (v, u) \in E \mid \chi(u) > \chi(v) } \label{eq:multipass:edge-parts}.
\end{aligned}
\end{equation}

For all $v \in V$ and $\alpha \geq 0$, we also define:\footnote{Note that when $\alpha = 0$, we have e.g. $\vY{\dirIn}{\alpha}{v} = \card*{ E_{ \dirIn, \dirHigh }(v)}$. The additional dependence of $\vY{\dirIn}{\alpha}{v}$ on $\card*{ E_{ \dirIn, \dirLow }(v)}$ when $\alpha > 0$ is necessary to accommodate for errors we accumulate when sampling edges from $E_{ \dirIn, \dirLow }(v)$ to estimate the $\vZ{\dirIn}{v}$ quantities we will see below, see \Cref{cor:cls17estimate:z}.}
\begin{equation}
\label{eq:multipass:y-value}
\begin{aligned}
\vY{\dirIn}{\alpha}{v} &\coloneqq \max\paren*{ \card*{ E_{ \dirIn, \dirHigh }(v) }, \alpha \cdot \card*{ E_{ \dirIn, \dirLow }(v) } } , \\
\vY{\dirOut}{\alpha}{v} &\coloneqq \max\paren*{ \card*{ E_{ \dirOut, \dirHigh }(v) }, \alpha \cdot \card*{ E_{ \dirOut, \dirLow }(v) } } .
\end{aligned}
\end{equation}
Observe that $\vY{\dirIn}{\alpha}{v}, \vY{\dirOut}{\alpha}{v} \geq 0$. Further,
\begin{equation}\label{eq:multipass:y-sum}
    \vY{\dirIn}{\alpha}{v} +  \vY{\dirOut}{\alpha}{v} > 0
\end{equation}
since there are no isolated vertices and also
\begin{equation}\label{eq:multipass:e-lb}
    \card*{ E_{\dirOut,\dirHigh}(v) } \geq \vY{\dirOut}{\alpha}{v} - \alpha \card*{ E_{\dirOut,\dirLow}(v) } \text{ and }
    \card*{ E_{\dirIn,\dirHigh}(v) } \geq \vY{\dirIn}{\alpha}{v} - \alpha \card*{ E_{\dirIn,\dirLow}(v) }.
\end{equation}

We are now ready to describe the algorithm.

\begin{algorithm}[H]
	\caption{Recursive, deterministic procedure to define fractional cut $\aPos$}
	\label{algo:cls17}
	\begin{algorithmic}[1]
        \renewcommand{\algorithmicrequire}{\textbf{Input:}}
		\renewcommand{\algorithmicensure}{\textbf{Output:}}
		
		\Require An integer $k > 0$, graph $G = \paren*{ V, E }$, a proper coloring $\chi : V \to [k]$ of $G$, $\alpha \geq 0$, and a vertex $v \in V$.
		\Ensure The fractional assignment $\aPos(v) \in [0,1]$ to $v$.

        \alglinenoPop{algcommon}
  
		\State Compute recursively: 
		\[
			\vZ{\dirIn}{v} = \sum_{ (u, v) \in E_{ \dirIn, \dirLow }(v) } \aPos(u) , \hspace{2cm} \vZ{\dirOut}{v} = \sum_{ (v, u) \in E_{ \dirOut, \dirLow }(v) } \paren*{ 1 - \aPos(u) } .
		\]\label{line:cls17:z}
		
		\State Output
		\[
			\begin{cases}
				1, &\text{~if~} \vZ{\dirIn}{v} - \vZ{\dirOut}{v} \leq - \vY{\dirIn}{\alpha}{v} \\
				\frac{ \vY{\dirOut}{\alpha}{v}}{ \vY{\dirIn}{\alpha}{v} + \vY{\dirOut}{\alpha}{v} } - \frac{\vZ{\dirIn}{v} - \vZ{\dirOut}{v} }{ \vY{\dirIn}{\alpha}{v} + \vY{\dirOut}{\alpha}{v} } , &\text{~if~} {-\vY{\dirIn}{\alpha}{v}} < \vZ{\dirIn}{v} - \vZ{\dirOut}{v} \leq \vY{\dirOut}{\alpha}{v} \\
				0, &\text{~if~} \vY{\dirOut}{\alpha}{v} < \vZ{\dirIn}{v} - \vZ{\dirOut}{v} 
			\end{cases} .
		\]\label{line:cls17:pos}

        \alglinenoPush{algcommon}
        
	\end{algorithmic}
\end{algorithm}

Note that while this algorithm is recursive, a call to $\aPos(v)$ only recurses into calls to $\aPos(u)$ when $\chi(u) < \chi(v)$ (since in $E_{\dirIn,\dirLow}(v)$ and $E_{\dirOut,\dirLow}(v)$, any non-$v$ vertex has lower color by definition (\Cref{eq:multipass:edge-parts})); thus, the recursion terminates at depth at most $k$. In particular, if $\chi(v) = 1$, there are no recursive calls at all, and we have $\vZ{\dirIn}{v} = \vZ{\dirOut}{v} = 0$. Also, the three cases in \Cref{line:cls17:pos} are mutually exclusive and exhaustive since \Cref{eq:multipass:y-sum} implies $-\vY{\dirIn}{\alpha}{v} < \vY{\dirOut}{\alpha}{v}$. Finally, we emphasize that this algorithm is deterministic.

We now set up some additional notation that will be useful in the sequel. For $v \in V$, we define the quantities:
\begin{equation} \label{eq:multipass:a}
\begin{aligned}
\vA{\dirIn}{\alpha}{v} &\coloneqq \vY{\dirIn}{\alpha}{v} + \vZ{\dirIn}{v} - \vZ{\dirOut}{v}, \\
\vA{\dirOut}{\alpha}{v} &\coloneqq \vY{\dirOut}{\alpha}{v} + \vZ{\dirOut}{v} - \vZ{\dirIn}{v} .
\end{aligned}
\end{equation}
Importantly,
\begin{equation}\label{eq:multipass:a-sum}
\vA{\dirIn}{\alpha}{v} + \vA{\dirOut}{\alpha}{v} = \vY{\dirIn}{\alpha}{v} + \vY{\dirOut}{\alpha}{v} > 0.
\end{equation}
Using this notation, we can rewrite \Cref{line:cls17:pos} as:
\begin{equation}\label{eq:multipass:cls17:pos-alt}
\aPos(v) =
\begin{cases}
    1, &\text{~if~} \vA{\dirIn}{\alpha}{v} \leq 0 \\
    \frac{ \vA{\dirOut}{\alpha}{v} }{ \vA{\dirIn}{\alpha}{v} + \vA{\dirOut}{\alpha}{v} } , &\text{~if~} \vA{\dirIn}{\alpha}{v} > 0 \text{ and } \vA{\dirOut}{\alpha}{v} \geq 0 \\
    0, &\text{~if~} \vA{\dirOut}{\alpha}{v} < 0
\end{cases}.
\end{equation}
Note that these cases correspond to the respective cases in \Cref{line:cls17:pos}. As a sanity check, we observe that $\vA{\dirIn}{\alpha}{v} \leq 0$ and $\vA{\dirOut}{\alpha}{v} < 0$ cannot hold simultaneously since the sum of the left-hand sides sum is positive (\Cref{eq:multipass:a-sum}). The upshot of \Cref{eq:multipass:cls17:pos-alt} is that $\aPos(v)$ is the ``truncation'' of the ratio $\frac{ \vA{\dirOut}{\alpha}{v} }{ \vA{\dirIn}{\alpha}{v} + \vA{\dirOut}{\alpha}{v} }$ to the $[0,1]$ interval (e.g., $\frac{ \vA{\dirOut}{\alpha}{v} }{ \vA{\dirIn}{\alpha}{v} + \vA{\dirOut}{\alpha}{v} } \geq 1$ iff $\vA{\dirIn}{\alpha}{v} \leq 0$ iff $\aPos(v) = 1$).

\subsubsection{Proof of correctness}

The goal of this subsection is to prove that \cref{algo:cls17} works, in the following sense:

\begin{lemma}
\label{lemma:multipass:cls17} 
Let $k > 0$, $G = \paren*{ V, E }$ be a graph, and $\chi : V \to [k]$ be a proper coloring of $G$. Let $\alpha \geq 0$ be arbitrary and $\aPos : V \to [0,1]$ be as defined by \cref{algo:cls17}. We have:
\[
\paren*{ \frac{1}{2} - \alpha } \cdot \maxval{G} \leq \val{G}{ \aPos } \leq \maxval{G} .
\]
\end{lemma}

Note that at $\alpha=0$, this implies that the fractional assignment $\aPos$ gives a $\frac12$-approximation for the $\mdcut$ value of $G$.

Fix $k$, $G$, $\chi$, and $\alpha$ as in the lemma statement and let $n = \card*{ V }$. Let $\aPos$ and $(\vZ{\dirOut}{w})_{w \in V},(\vZ{\dirIn}{w})_{w \in V}$ denote the result of running \Cref{algo:cls17}.

We first define some machinery that will be useful in proving \cref{lemma:multipass:cls17}. Fix an arbitrary bijection $\pi : V \to [n]$ that satisfies that for all $u, v \in V$, $\pi(u) < \pi(v) \implies \chi(u) < \chi(v)$. For all assignments $\aX : V \to \bits$ and all $0 \leq j \leq n$, we define the fractional assignment $\aPiece_{ j, \aX } : V \to [0, 1]$ as follows:
\begin{equation}
\label{eq:multipass:piece}
\aPiece_{ j, \aX }(v) \coloneqq
\begin{cases}
\aPos(v) , &\text{~if~} \pi(v) \leq j \\
\aX(v) , &\text{~if~} \pi(v) > j
\end{cases} .
\end{equation}
I.e., $\aPiece_{j,\aX}$ is a fractional assignment matching the assignment $\aPos$ on vertices $\pi^{-1}(1),\ldots,\pi^{-1}(j)$ and $\aX$ on the remaining vertices. Note that $\aPiece_{0,\aX} = \aX$ and $\aPiece_{n,\aX} = \aPos$.

The following lemma bounds the loss in value from $\aPiece_{j-1,\aX}$ to $\aPiece_{j,\aX}$, i.e., from changing vertex $\pi^{-1}(j)$'s assignment from $\aX(j)$ to $\aPos(\pi^{-1}(j))$ while maintaining the rest of the assignment ($\aPos$ on vertices $\pi^{-1}(1),\ldots,\pi^{-1}(j-1)$ and $\aX$ on the remaining vertices).

\begin{lemma}
\label{lemma:multipass:piece}
Let $\aX : V \to \bits$, $j \in [n]$, and $w = \pi^{-1}(j)$. We have:
\begin{multline*}
\val{G}{ \aPiece_{ j - 1, \aX } } - \val{G}{ \aPiece_{ j, \aX } } \\
= \frac{ \aX(w) - \aPos(w) }{ \card*{ E } } \cdot \paren*{ \vZ{\dirOut}{w} - \vZ{\dirIn}{w} + \sum_{ (w, v) \in E_{ \dirOut, \dirHigh }(w) } \paren*{ 1 - \aX(v) } - \sum_{ (u, w) \in E_{ \dirIn, \dirHigh }(w) } \aX(u) } .
\end{multline*}
\end{lemma}

\begin{proof}
From \cref{eq:val}, we get:
\begin{align*}
&\val{G}{ \aPiece_{ j - 1, \aX } } - \val{G}{ \aPiece_{ j, \aX } } \\
&= \frac{1}{ \card*{ E } } \cdot \sum_{ (u, v) \in E } \paren*{ \aPiece_{ j - 1, \aX }(u) \cdot \paren*{ 1 - \aPiece_{ j - 1, \aX }(v) } - \aPiece_{ j, \aX }(u) \cdot \paren*{ 1 - \aPiece_{ j, \aX }(v) } } \\
\intertext{As $\pi$ is a bijection, for all $v' \neq w$, we have $\pi\paren*{ v' } \neq j$. From \cref{eq:multipass:piece}, we get that $\aPiece_{ j - 1, \aX }\paren*{ v' } = \aPiece_{ j, \aX }\paren*{ v' }$ for all $v' \neq w$. It follows that any term above where both $u, v \neq w$ vanishes. We are left with terms where either $u = w \neq v$ or $u \neq w = v$ (note that $u = w = v$ is impossible as $\chi$ is a proper coloring implying that there are no self-loops). Using $\aPiece_{ j - 1, \aX }\paren*{ v' } = \aPiece_{ j, \aX }\paren*{ v' }$ for all $v' \neq w$, we rearrange to:}
&= \frac{1}{ \card*{ E } } \cdot \paren*{ \aPiece_{ j - 1, \aX }(w) - \aPiece_{ j, \aX }(w) } \cdot \paren*{ \sum_{ (w, v) \in E } \paren*{ 1 - \aPiece_{ j, \aX }(v) } - \sum_{ (u, w) \in E } \aPiece_{ j, \aX }(u) } .
\intertext{We now use the definition of $\aPiece$ (\cref{eq:multipass:piece}) and the fact that $v' \neq w \implies \pi(v') \neq j$ in all the terms (no self-loops) to get:}
&= \frac{ \aX(w) - \aPos(w) }{ \card*{ E } } \cdot \paren*{ \sum_{ \substack{ (w, v) \in E \\ \pi(v) < j } } \paren*{ 1 - \aPos(v) } - \sum_{ \substack{ (u, w) \in E \\ \pi(u) < j } } \aPos(u) + \sum_{ \substack{ (w, v) \in E \\ \pi(v) > j } } \paren*{ 1 - \aX(v) } - \sum_{ \substack{ (u, w) \in E \\ \pi(u) > j } } \aX(u) } .
\intertext{As $\chi$ is a proper coloring and $\pi$ is a bijection, for all edges $(u, v) \in E$, we have $\chi(u) \neq \chi(v) \implies u \neq v \iff \pi(u) \neq \pi(v)$. From the definition of $\pi$, this means that for all edges $(u, v) \in E$, we have $\pi(u) < \pi(v) \iff \chi(u) < \chi(v)$ and $\pi(u) > \pi(v) \iff \chi(u) > \chi(v)$. Plugging in and using the definition of the edge partitions (\cref{eq:multipass:edge-parts}) and the $\vZ{\cdot}{\cdot}$ variables (\cref{line:cls17:z}), this gives:}
&= \frac{ \aX(w) - \aPos(w) }{ \card*{ E } } \cdot \paren*{ \vZ{\dirOut}{w} - \vZ{\dirIn}{w} + \sum_{ (w, u) \in E_{ \dirOut, \dirHigh }(w) } \paren*{ 1 - \aX(u) } - \sum_{ (u, w) \in E_{ \dirIn, \dirHigh }(w) } \aX(u) } ,
\end{align*}
as desired.
\end{proof}

Note that the proof of the previous lemma did not depend on the definition of the assignment $\aPos$.

For the rest of the proof, define $\aOpt : V \to \bits$ be the maximizer in \cref{eq:maxval} so that we have $\maxval{G} = \val{G}{ \aOpt }$. The following lemma specializes the previous lemma to the case $\aX = \aOpt$:
\begin{lemma}
\label{lemma:multipass:piece-opt}
Let $j \in [n]$ and $w = \pi^{-1}(j)$. If either $\vA{\dirIn}{\alpha}{w} \leq 0$ or $\vA{\dirOut}{\alpha}{w} \leq 0$, we have $\val{G}{ \aPiece_{ j - 1, \aOpt } } - \val{G}{ \aPiece_{ j, \aOpt } } \leq 0$. Otherwise:
\[
\val{G}{ \aPiece_{ j - 1, \aOpt } } - \val{G}{ \aPiece_{ j, \aOpt } } \leq \frac{1}{ \card*{ E } } \cdot \frac{ \vA{\dirIn}{\alpha}{w} \cdot \vA{\dirOut}{\alpha}{w} }{ \vA{\dirIn}{\alpha}{w} + \vA{\dirOut}{\alpha}{w} } .
\]
\end{lemma}

\begin{proof}
Setting $\aX = \aOpt$ in \cref{lemma:multipass:piece}, we get:
\begin{multline*}
\val{G}{ \aPiece_{ j - 1, \aOpt } } - \val{G}{ \aPiece_{ j, \aOpt } } \\
= \frac{ \aOpt(w) - \aPos(w) }{ \card*{ E } } \cdot \paren*{ \vZ{\dirOut}{w} - \vZ{\dirIn}{w} + \sum_{ (w, u) \in E_{ \dirOut, \dirHigh }(w) } \paren*{ 1 - \aOpt(u) } - \sum_{ (u, w) \in E_{ \dirIn, \dirHigh }(w) } \aOpt(u) } .
\end{multline*}
Consider now the following cases:
\begin{itemize}
\item \textbf{When $\aOpt(w) = 0$:} Plugging this in and using $\aOpt(u) \in [0,1]$, we get:
\[
\val{G}{ \aPiece_{ j - 1, \aOpt } } - \val{G}{ \aPiece_{ j, \aOpt } } \leq \frac{1}{ \card*{ E } } \cdot \aPos(w) \cdot \paren*{ \card*{ E_{ \dirIn, \dirHigh }(w) } + \vZ{\dirIn}{w} - \vZ{\dirOut}{w} } .
\]
From the definitions of $\vY{\dirIn}{\alpha}{w}$ (\cref{eq:multipass:y-value}) and $\vA{\dirIn}{\alpha}{w}$ (\cref{eq:multipass:a}), we get:
\begin{align*}
\val{G}{ \aPiece_{ j - 1, \aOpt } } - \val{G}{ \aPiece_{ j, \aOpt } } &\leq \frac{1}{ \card*{ E } } \cdot \aPos(w) \cdot \paren*{ \vY{\dirIn}{\alpha}{w} + \vZ{\dirIn}{w} - \vZ{\dirOut}{w} } \\
&= \frac{1}{ \card*{ E } } \cdot \aPos(w) \cdot \vA{\dirIn}{\alpha}{w} .
\end{align*}
If $\vA{\dirIn}{\alpha}{w} \leq 0$, the lemma follows since $\aPos(w) \geq 0$. Else, from the definition of $\aPos$ (\cref{eq:multipass:cls17:pos-alt}), we have $\aPos(w) = 0$ if $\vA{\dirOut}{\alpha}{w} \leq 0$, in which case the lemma again follows. Otherwise, $\vA{\dirIn}{\alpha}{w}$ and $\vA{\dirOut}{\alpha}{w}$ are both positive, and by \cref{eq:multipass:cls17:pos-alt} \[ \aPos(w) = \frac{ \vA{\dirOut}{\alpha}{w} }{ \vA{\dirIn}{\alpha}{w} + \vA{\dirOut}{\alpha}{w} },
\]
in which case the lemma again follows.
\item \textbf{When $\aOpt(w) = 1$:} Plugging this in, we get:
\[
\val{G}{ \aPiece_{ j - 1, \aOpt } } - \val{G}{ \aPiece_{ j, \aOpt } } \leq \frac{1}{ \card*{ E } } \cdot \paren*{ 1 - \aPos(w) } \cdot \paren*{ \card*{ E_{ \dirOut, \dirHigh }(w) } + \vZ{\dirOut}{w} - \vZ{\dirIn}{w} } .
\]
From the definitions of $\vY{\dirOut}{\alpha}{w}$ (\cref{eq:multipass:y-value}) and $\vA{\dirOut}{\alpha}{w}$ (\cref{eq:multipass:a}), we get:
\begin{align*}
\val{G}{ \aPiece_{ j - 1, \aOpt } } - \val{G}{ \aPiece_{ j, \aOpt } } &\leq \frac{1}{ \card*{ E } } \cdot \paren*{ 1 - \aPos(w) } \cdot \paren*{ \vY{\dirOut}{\alpha}{w} + \vZ{\dirOut}{w} - \vZ{\dirIn}{w} } \\
&= \frac{1}{ \card*{ E } } \cdot \paren*{ 1 - \aPos(w) } \cdot \vA{\dirOut}{\alpha}{w} .
\end{align*}
We conclude using \cref{eq:multipass:cls17:pos-alt}, symmetrically to the previous case.
\end{itemize}
\end{proof}

Let $\aZero, \aOne : V \to \bits$ be the assignments that are $0$ and $1$ everywhere respectively.

\begin{lemma}\label{lemma:multipass:telescope-step}
    Let $j \in [n]$ and $w = \pi^{-1}(j)$. Then:
    \begin{multline*}
    \val{G}{ \aPiece_{ j, \aZero } } - \val{G}{ \aPiece_{ j - 1, \aZero } } + \val{G}{ \aPiece_{ j, \aOne } } - \val{G}{ \aPiece_{ j - 1, \aOne } } \\
    \geq 2 \cdot \paren*{ \val{G}{ \aPiece_{ j - 1, \aOpt } } - \val{G}{ \aPiece_{ j, \aOpt } } } - \frac{ \alpha \cdot \paren*{ \card*{ E_{ \dirOut, \dirLow }(w) } + \card*{ E_{ \dirIn, \dirLow }(w) } } }{ \card*{ E } } .
    \end{multline*}
\end{lemma}

\begin{proof}
We have by \cref{lemma:multipass:piece} that:
\begin{align*}
&\val{G}{ \aPiece_{ j, \aZero } } - \val{G}{ \aPiece_{ j - 1, \aZero } } + \val{G}{ \aPiece_{ j, \aOne } } - \val{G}{ \aPiece_{ j - 1, \aOne } } \\
&= \frac{1}{ \card*{ E } } \cdot \paren*{ \aPos(w) \cdot \paren*{ \card*{ E_{ \dirOut, \dirHigh }(w) } + \vZ{\dirOut}{w} - \vZ{\dirIn}{w} } + \paren*{ 1 - \aPos(w) } \cdot \paren*{ \card*{ E_{ \dirIn, \dirHigh }(w) } + \vZ{\dirIn}{w} - \vZ{\dirOut}{w} } }.
\intertext{Since $\aPos(w) \in [0, 1]$, by \cref{eq:multipass:e-lb,eq:multipass:a} we get:}
&\geq \frac{1}{ \card*{ E } } \cdot \paren*{ \aPos(w) \cdot \paren*{ \vA{\dirOut}{\alpha}{w} - \alpha \cdot \card*{ E_{ \dirOut, \dirLow }(w) } } + \paren*{ 1 - \aPos(w) } \cdot \paren*{ \vA{\dirIn}{\alpha}{w} - \alpha \cdot \card*{ E_{ \dirIn, \dirLow }(w) } } } \\
&= \frac{1}{ \card*{ E } } \cdot \paren*{ \aPos(w) \cdot \vA{\dirOut}{\alpha}{w} + \paren*{ 1 - \aPos(w) } \cdot \vA{\dirIn}{\alpha}{w} } - \frac{ \alpha \cdot \paren*{ \card*{ E_{ \dirOut, \dirLow }(w) } + \card*{ E_{ \dirIn, \dirLow }(w) } } }{ \card*{ E } }.
\intertext{To continue, note from \cref{eq:multipass:cls17:pos-alt} that $\vA{\dirIn}{\alpha}{w} \leq 0$ implies that $\aPos(w) = 1$ and that $\vA{\dirOut}{\alpha}{w} \leq 0$ implies that $\aPos(w) = 0$. This means that the first term above is always non-negative. We claim that this term is at least $2 \cdot \paren*{ \val{G}{ \aPiece_{ j - 1, \aOpt } } - \val{G}{ \aPiece_{ j, \aOpt } } }$. If either $\vA{\dirIn}{\alpha}{w} \leq 0$ or $\vA{\dirOut}{\alpha}{w} \leq 0$, this follows from \cref{lemma:multipass:piece-opt} and the fact that the term is non-negative. Otherwise, we have from \cref{eq:multipass:cls17:pos-alt} that $\aPos(w) = \frac{ \vA{\dirOut}{\alpha}{w} }{ \vA{\dirIn}{\alpha}{w} + \vA{\dirOut}{\alpha}{w} }$ and the claim follows from \cref{lemma:multipass:piece-opt} and the fact that, for all $x, y \in \mathbb{R}$, we have $x^2 + y^2 \geq 2xy$. With this, we get:}
&\geq 2 \cdot \paren*{ \val{G}{ \aPiece_{ j - 1, \aOpt } } - \val{G}{ \aPiece_{ j, \aOpt } } } - \frac{ \alpha \cdot \paren*{ \card*{ E_{ \dirOut, \dirLow }(w) } + \card*{ E_{ \dirIn, \dirLow }(w) } } }{ \card*{ E } },
\end{align*}
as desired.
\end{proof}

We are finally ready to prove \cref{lemma:multipass:cls17}.

\begin{proof}[Proof of \cref{lemma:multipass:cls17}]
We prove the first inequality; the second inequality is trivial since $\maxval{G}$ is the maximum value over all assignments (including fractional assignments WLOG). Recall that \Cref{lemma:multipass:telescope-step} is true for all $j \in [n]$. Summing it up for all $j \in [n]$, we get:
\begin{align*}
&\val{G}{ \aPiece_{ n, \aZero } } - \val{G}{ \aPiece_{ 0, \aZero } } + \val{G}{ \aPiece_{ n, \aOne } } - \val{G}{ \aPiece_{ 0, \aOne } } \\
&\geq 2 \cdot \paren*{ \val{G}{ \aPiece_{ 0, \aOpt } } - \val{G}{ \aPiece_{ n, \aOpt } } } - \sum_{ w \in V } \frac{ \alpha \cdot \paren*{ \card*{ E_{ \dirOut, \dirLow }(w) } + \card*{ E_{ \dirIn, \dirLow }(w) } } }{ \card*{ E } } .
\intertext{As \cref{eq:multipass:edge-parts} holds and $\chi$ is a proper coloring, the multisets $\cup_{w \in V} \paren*{ E_{ \dirOut, \dirLow }(w) }$ and $\cup_{w \in V} \paren*{ E_{ \dirIn, \dirLow }(w) }$ together form a partition of $E$. Using this, we get:}
&\geq 2 \cdot \paren*{ \val{G}{ \aPiece_{ 0, \aOpt } } - \val{G}{ \aPiece_{ n, \aOpt } } } - \alpha .
\end{align*}
Using the definition of $\aPiece$ (\Cref{eq:multipass:piece}, in particular that $\aPiece_{0,\aX} = \aX$ and $\aPiece_{n,\aX} = \aPos$ for all $\aX$) and the fact that $\aZero$ and $\aOne$ have value zero, we get:
\[
4 \cdot \val{G}{ \aPos } \geq 2 \cdot \maxval{G} - \alpha .
\]
We are done as $\maxval{G} \geq \frac{1}{4}$.
\end{proof}

\subsection{Estimating $\aPos(\cdot)$}
\label{sec:multipass:cls17estimate}

We now give an algorithm to estimate the function $\aPos(\cdot)$ computed in \cref{algo:cls17}. The main idea is to use random sampling to estimate the sums in \cref{line:cls17:z} instead of computing them exactly. Similar to \cref{algo:cls17}, \cref{algo:cls17estimate} is not in the streaming setting and its input is a graph $G = \paren*{ V, E }$, a proper coloring $\chi : V \to [k]$ of the vertices in $G$ with $k$ colors (for some $k > 0$), and a vertex $v \in V$ whose $\aPos(v)$ value we are trying to estimate. Unlike \cref{algo:cls17}, we henceforth fix $\alpha = \epsilon^5$ instead of working with a general $\alpha$ and omit it from the notation above. To start, define the parameters:
\begin{equation}
\label{eq:multipass:params}
D = \epsilon^{ - 100k } \hspace{1cm}\text{and}\hspace{1cm} \forall a \in [k] : \delta_a = \epsilon^{ 10 \cdot \paren*{ k + 1 - a } } .
\end{equation}
We will eventually use $\delta_a$ to bound ``errors'' for color-$a$ vertices.
\begin{equation}\label{eq:multipass:delta-k-bound}
    \delta_k = \epsilon^{10} < \epsilon^5 = \alpha.
\end{equation}
We also use the simple estimate that for all $x > 0$,
\begin{equation}\label{eq:multipass:exp-est}
    \mathrm{e}^{-x^{-1}} \leq x.
\end{equation}
(This can be proven using that $\mathrm{e}^y \geq y$ for all $y \in \BR$, setting $y=x^{-1}$, and reciprocating.)

\begin{algorithm}[H]
	\caption{Recursive, randomized procedure to sample estimate $\aPosEst(v)$ of $\aPos(v)$.}
	\label{algo:cls17estimate}
	\begin{algorithmic}[1]
		\renewcommand{\algorithmicrequire}{\textbf{Input:}}
		\renewcommand{\algorithmicensure}{\textbf{Output:}}
		
		\Require An integer $k > 0$, a graph $G = \paren*{ V, E }$, a proper coloring $\chi : V \to [k]$ of $G$, and a vertex $v \in V$.
		\Ensure A (random) value $\aPosEst(v) \in [0,1]$.

        \alglinenoPop{algcommon}
        
		\State Set $\vZEst{\dirIn}{v} = 0$ if $E_{ \dirIn, \dirLow }(v) = \emptyset$. If not, sample independent $\paren*{ \mathsf{u}_1, v }, \dots, \paren*{ \mathsf{u}_D, v } \sim \Unif{E_{ \dirIn, \dirLow }(v)}$ and set recursively:
		\[
			\vZEst{\dirIn}{v} = \frac{ \card*{ E_{ \dirIn, \dirLow }(v) } }{D} \cdot \sum_{d = 1}^D \aPosEst\paren*{ \mathsf{u}_d } .
		\]
		Similarly, we set $\vZEst{\dirOut}{v} = 0$ if $E_{ \dirOut, \dirLow }(v) = \emptyset$. Otherwise, we sample independent $\paren*{ v, \mathsf{u}_1 }, \dots, \paren*{ v, \mathsf{u}_D } \sim \Unif{E_{ \dirOut, \dirLow }(v)}$ and set recursively:
		\[
			\vZEst{\dirOut}{v} = \frac{ \card*{ E_{ \dirOut, \dirLow }(v) } }{D} \cdot \sum_{d = 1}^D \paren*{ 1 - \aPosEst\paren*{ \mathsf{u}_d } } .
		\]
        We emphasize that each invocation of $\aPosEst(\cdot)$ above uses fresh randomness. \label{line:cls17estimate:z}
		
		\State Output \label{line:cls17estimate:pos}
		\[
			\begin{cases}
				1, &\text{~if~} \vZEst{\dirIn}{v} - \vZEst{\dirOut}{v} \leq - \vYNoAlpha{\dirIn}{v} \\
				\frac{ \vYNoAlpha{\dirOut}{v} }{ \vYNoAlpha{\dirIn}{v} + \vYNoAlpha{\dirOut}{v} } - \frac{ \vZEst{\dirIn}{v} - \vZEst{\dirOut}{v} }{ \vYNoAlpha{\dirIn}{v} + \vYNoAlpha{\dirOut}{v} }, &\text{~if~} {- \vYNoAlpha{\dirIn}{v}} < \vZEst{\dirIn}{v} - \vZEst{\dirOut}{v} \leq \vYNoAlpha{\dirOut}{v} \\
				0, &\text{~if~} \vYNoAlpha{\dirOut}{v} < \vZEst{\dirIn}{v} - \vZEst{\dirOut}{v}
			\end{cases} .
		\]

        \alglinenoPush{algcommon}
        
	\end{algorithmic}
\end{algorithm}

We show that \cref{algo:cls17estimate} indeed estimates $\aPos(\cdot)$.

\begin{lemma}
\label{lemma:multipass:cls17estimate} 
Let $k > 0$, $G = \paren*{ V, E }$ be a graph, and $\chi : V \to [k]$ be a proper coloring of $G$. Let $v \in V$, and let $\aPos(v)$ and $\aPosEst(v)$ be as defined by \cref{algo:cls17,algo:cls17estimate}, respectively. Then:
\[
\Pr\paren*{ \abs*{ \aPosEst(v) - \aPos(v) } \geq \delta_{ \chi(v) } } \leq \delta_{ \chi(v) } .
\]
\end{lemma}
Before proving \cref{lemma:multipass:cls17estimate}, we quickly note that \cref{claim:cls17exp-helper} and \cref{lemma:multipass:cls17estimate} together imply that $\Exp \aPosEst(v)$ is close to $\aPos(v)$.

\begin{proof}
Fix $k$, $G$, and $\chi$ as in the lemma statement. We will show the lemma by induction on $\chi(v)$. For the base case, consider any vertex $v$ satisfying $\chi(v) = 1$. For such a vertex, \cref{eq:multipass:edge-parts} says that $E_{ \dirIn, \dirLow }(v) = E_{ \dirOut, \dirLow }(v) = \emptyset$. Plugging into \cref{line:cls17:z,line:cls17estimate:z}, we get that $\vZ{\dirIn}{v} = \vZ{\dirOut}{v} = 0$ and $\vZEst{\dirIn}{v} = \vZEst{\dirOut}{v} = 0$. From \cref{line:cls17:pos,line:cls17estimate:pos}, we get that $\aPosEst(v) = \aPos(v)$ with probability $1$ and the lemma follows trivially.

For the inductive step, consider $a > 1$ and a vertex $v$ satisfying $\chi(v) = a$. We now prove the result for $v$ assuming the following ``induction hypothesis'':
\begin{equation}\label{eq:multipass:ind-hyp}
    \forall u \in V, \chi(u) < a : \quad \Pr\paren*{ \abs*{ \aPosEst(u) - \aPos(u) } \geq \delta_{ \chi(u) } } \leq \delta_{ \chi(u) } .
\end{equation}

Firstly, we observe:
\begin{claim}
\label{claim:multipass:cls17exp}
For all vertices $u$ satisfying $\chi(u) < a$, we have:
\[
\abs*{ \Exp\bracket*{ \aPosEst(u) } - \aPos(u) } \leq 2 \delta_{ a - 1 } .
\]
\end{claim}
\begin{proof}
Follows from the induction hypothesis (\Cref{eq:multipass:ind-hyp}) and \cref{claim:cls17exp-helper}.
\end{proof}

Next, we state a concentration result for sums of $\aPos$-values over randomly sampled in- or out-neighbors:
\begin{claim}
\label{claim:multipass:cls17estimate:chernoff}
If $E_{ \dirIn, \dirLow }(v) \neq \emptyset$, we have:
\[
\Pr_{ \paren*{ \mathsf{u}_1, v }, \dots, \paren*{ \mathsf{u}_D, v } \sim E_{ \dirIn, \dirLow }(v) }\paren*{ \abs*{ \sum_{d = 1}^D \aPos\paren*{ \mathsf{u}_d } - \frac{D}{ \card*{ E_{ \dirIn, \dirLow }(v) } } \cdot \sum_{ (u, v) \in E_{ \dirIn, \dirLow }(v) } \aPos(u) } \geq D \cdot \delta_{ a - 1 } } \leq \delta_{ a - 1 } .
\]
Similarly, if $E_{ \dirOut, \dirLow }(v) \neq \emptyset$, we have:
\[
\Pr_{ \paren*{ v, \mathsf{u}_1 }, \dots, \paren*{ v, \mathsf{u}_D } \sim E_{ \dirOut, \dirLow }(v) }\paren*{ \abs*{ \sum_{d = 1}^D \paren*{ 1 - \aPos\paren*{ \mathsf{u}_d } } - \frac{D}{ \card*{ E_{ \dirOut, \dirLow }(v) } } \cdot \sum_{ (v, u) \in E_{ \dirOut, \dirLow }(v) } \paren*{ 1 - \aPos(u) } } \geq D \cdot \delta_{ a - 1 } } \leq \delta_{ a - 1 } .
\]
(Edges are sampled independently and uniformly at random.)
\end{claim}
\begin{proof}
The proof is a standard Chernoff bound. We only show the first inequality as the second one is analogous. For $d \in [D]$, let $\mathsf{X}_d$ be a random variable that takes the value $\aPos\paren*{ \mathsf{u}_d }$. Observe that these random variables are mutually independent and identically distributed with expectation $\frac{1}{ \card*{ E_{ \dirIn, \dirLow }(v) } } \cdot \sum_{ (u, v) \in E_{ \dirIn, \dirLow }(v) } \aPos(u)$. Moreover, they only take values in the interval $[0, 1]$. From the Chernoff bound (\cref{prop:hoeffding}) we have:
\[
\Pr_{ \paren*{ \mathsf{u}_1, v }, \dots, \paren*{ \mathsf{u}_D, v } \sim E_{ \dirIn, \dirLow }(v) }\paren*{ \abs*{ \sum_{d = 1}^D \aPos\paren*{ \mathsf{u}_d } - \frac{D}{ \card*{ E_{ \dirIn, \dirLow }(v) } } \cdot \sum_{ (u, v) \in E_{ \dirIn, \dirLow }(v) } \aPos(u) } \geq D \cdot \delta_{ a - 1 } } \leq 2 \cdot \mathrm{e}^{ - 2D \cdot \delta_{ a - 1 }^2 } \] which by our definition of $\delta_{a-1}$ (\Cref{eq:multipass:params}) and \Cref{eq:multipass:exp-est} is at most $\delta_{ a - 1}$.
\end{proof}

Recall the value $\vZ{\dirIn}{v}$ computed in \cref{line:cls17:z} of \Cref{algo:cls17}: If $E_{ \dirIn, \dirLow }(v) \neq \emptyset$, $\vZ{\dirIn}{v}$ takes the average $\aPos(u)$ value over neighbors from $E_{ \dirIn, \dirLow }(v)$ edges; by comparison, $\vZEst{\dirIn}{v}$ samples $D$ random neighbors $\mathsf{u}_1,\ldots,\mathsf{u}_D$ from $E_{ \dirIn, \dirLow }(v)$ edges, samples an estimate $p_i \coloneqq \aPosEst(\mathsf{u}_i)$ for each, and then outputs the average of these $p_i$'s. (If $E_{ \dirIn, \dirLow }(v) = \emptyset$, then $\vZ{\dirIn}{v} = \vZEst{\dirIn}{v} = 0$.) The following lemma shows that $\vZEst{\dirIn}{v}$ is likely close to $\vZ{\dirIn}{v}$, and similarly for $\vZ{\dirOut}{v}$ (also computed in \cref{line:cls17:z}) and $\vZEst{\dirOut}{v}$. In proving the lemma, we have to account for three sources of error: Deviation between the average $\aPos$ value over the random neighbors $\mathsf{u}_1,\ldots,\mathsf{u}_D$ and the average $\aPos$ value over all neighbors of $v$; deviation between an expected estimate $\Exp \aPosEst(\mathsf{u}_i)$ and $\aPos(\mathsf{u}_i)$; and deviation between the average of the $p_i$'s and the average of the $\Exp \aPosEst(\mathsf{u}_i)$'s.

\begin{claim}
\label{claim:multipass:cls17estimate:z}
We have:
\begin{align*}
\Pr\paren*{ \abs*{ \vZEst{\dirIn}{v} - \vZ{\dirIn}{v} } > 10 \cdot \card*{ E_{ \dirIn, \dirLow }(v) } \cdot \delta_{ a - 1 } } &\leq 2 \delta_{ a - 1 } , \\
\Pr\paren*{ \abs*{ \vZEst{\dirOut}{v} - \vZ{\dirOut}{v} } > 10 \cdot \card*{ E_{ \dirOut, \dirLow }(v) } \cdot \delta_{ a - 1 } } &\leq 2 \delta_{ a - 1 } ,
\end{align*}
where the probabilities are both over the random choices of neighbors and the invocations of $\aPosEst$.
\end{claim}
\begin{proof}
We only show the first inequality as the second one is similar. Observe from \cref{line:cls17:z,line:cls17estimate:z} that if $E_{ \dirIn, \dirLow }(v) = \emptyset$, the probability under consideration is $0$ and there is nothing to show. Thus, throughout the proof we assume that $E_{ \dirIn, \dirLow }(v) \neq \emptyset$. Using the definition of $\vZEst{\dirIn}{v}$ and $\vZ{\dirIn}{v}$, we get:
\begin{align*}
&\Pr\paren*{ \abs*{ \vZEst{\dirIn}{v} - \vZ{\dirIn}{v} } > 10 \cdot \card*{ E_{ \dirIn, \dirLow }(v) } \cdot \delta_{ a - 1 } } \\
&= \Pr_{ \paren*{ \mathsf{u}_1, v }, \dots, \paren*{ \mathsf{u}_D, v } \sim E_{ \dirIn, \dirLow }(v) }\paren*{ \abs*{ \sum_{d = 1}^D \aPosEst\paren*{ \mathsf{u}_d } - \frac{D}{ \card*{ E_{ \dirIn, \dirLow }(v) } } \cdot \sum_{ (u, v) \in E_{ \dirIn, \dirLow }(v) } \aPos(u) } > 10D \cdot \delta_{ a - 1 } } 
\intertext{which, letting $\mathcal{E}_1$ be the event whose probability is upper bounded in \cref{claim:multipass:cls17estimate:chernoff}, by a chain rule, gives:}
&\leq \Pr\paren*{ \mathcal{E}_1 } + \Pr_{ \paren*{ \mathsf{u}_1, v }, \dots, \paren*{ \mathsf{u}_D, v } \sim E_{ \dirIn, \dirLow }(v) }\paren*{ \abs*{ \sum_{d = 1}^D \aPosEst\paren*{ \mathsf{u}_d } - \frac{D}{ \card*{ E_{ \dirIn, \dirLow }(v) } } \cdot \sum_{ (u, v) \in E_{ \dirIn, \dirLow }(v) } \aPos(u) } > 10D \cdot \delta_{ a - 1 } ~\bigg\vert~ \overline{ \mathcal{E}_1 } } .
\end{align*}
By \cref{claim:multipass:cls17estimate:chernoff}, it suffices to bound the second term by $\delta_{ a - 1 }$. We shall show this bound under a stronger conditioning by fixing an arbitrary $\paren*{ u_1, v }, \dots, \paren*{ u_D, v } \in E_{ \dirIn, \dirLow }(v)$ for which $\mathcal{E}_1$ does not occur. Observe that, after this stronger conditioning, the only randomness left is the randomness in the computations of $\aPosEst(\cdot)$. For an arbitrary such $\paren*{ u_1, v }, \dots, \paren*{ u_D, v } \in E_{ \dirIn, \dirLow }(v)$, we have:
\begin{align*}
&\Pr\paren*{ \abs*{ \sum_{d = 1}^D \aPosEst\paren*{ u_d } - \frac{D}{ \card*{ E_{ \dirIn, \dirLow }(v) } } \cdot \sum_{ (u, v) \in E_{ \dirIn, \dirLow }(v) } \aPos(u) } > 10D \cdot \delta_{ a - 1 } ~\bigg\vert~ \paren*{ u_1, v }, \dots, \paren*{ u_D, v } } \\
&= \Pr\paren*{ \abs*{ \sum_{d = 1}^D \aPosEst\paren*{ u_d } - \frac{D}{ \card*{ E_{ \dirIn, \dirLow }(v) } } \cdot \sum_{ (u, v) \in E_{ \dirIn, \dirLow }(v) } \aPos(u) } > 10D \cdot \delta_{ a - 1 } } ,
\intertext{as the randomness in the computations of $\aPosEst(\cdot)$ is independent of the randomness used to sample $\paren*{ \mathsf{u}_1, v }, \dots, \paren*{ \mathsf{u}_D, v } \sim E_{ \dirIn, \dirLow }(v)$. As $\paren*{ u_1, v }, \dots, \paren*{ u_D, v } \in E_{ \dirIn, \dirLow }(v)$ are chosen so that $\mathcal{E}_1$ does not occur, we have from \cref{claim:multipass:cls17estimate:chernoff} that $\abs*{ \sum_{d = 1}^D \aPos\paren*{ u_d } - \frac{D}{ \card*{ E_{ \dirIn, \dirLow }(v) } } \cdot \sum_{ (u, v) \in E_{ \dirIn, \dirLow }(v) } \aPos(u) } \leq D \cdot \delta_{ a - 1 }$. By a triangle inequality, we can continue as:}
&\leq \Pr\paren*{ \abs*{ \sum_{d = 1}^D \aPosEst\paren*{ u_d } - \sum_{d = 1}^D \aPos\paren*{ u_d } } \geq 8D \cdot \delta_{ a - 1 } } .
\intertext{We now apply \cref{claim:multipass:cls17exp} on $u_1, \dots, u_D$ to get that for all $d \in [D]$, we have $\abs*{ \Exp\bracket*{ \aPosEst\paren*{ u_d } } - \aPos\paren*{ u_d } } \leq 2 \delta_{ a - 1 }$. By a triangle inequality, this means that we have $\abs*{ \sum_{d = 1}^D \Exp\bracket*{ \aPosEst\paren*{ u_d } } - \sum_{d = 1}^D \aPos\paren*{ u_d } } \leq 2D \cdot \delta_{ a - 1 }$. By another triangle inequality, we get:}
&\leq \Pr\paren*{ \abs*{ \sum_{d = 1}^D \aPosEst\paren*{ u_d } - \sum_{d = 1}^D \Exp\bracket*{ \aPosEst\paren*{ u_d } } } \geq 5D \cdot \delta_{ a - 1 } } .
\intertext{We now bound the last term using \cref{prop:hoeffding,eq:multipass:params}. Observe that the random variables $\aPosEst\paren*{ u_1 }, \dots, \aPosEst\paren*{ u_D }$ are mutually independent and takes values in $[0, 1]$. We get using \Cref{eq:multipass:exp-est} that:}
&\leq 2 \cdot \mathrm{e}^{ - 50D \cdot \delta_{ a - 1 }^2 } \leq \delta_{ a - 1 },
\end{align*}
as desired.
\end{proof}
A union bound now gives:
\begin{corollary}
\label{cor:cls17estimate:z}
We have:
\[
\Pr\paren[\bigg]{ \abs*{ \paren*{ \vZEst{\dirIn}{v} - \vZEst{\dirOut}{v} } - \paren*{ \vZ{\dirIn}{v} - \vZ{\dirOut}{v} } } > 10 \delta_{ a - 1 } \cdot \paren*{ \card*{ E_{ \dirIn, \dirLow }(v) } + \card*{ E_{ \dirOut, \dirLow }(v) } } } \leq 4 \delta_{ a - 1 } .
\]
\end{corollary}

We now continue with our proof of \cref{lemma:multipass:cls17estimate}. For this, define the function $\cutoff_v : \mathbb{R} \to [0, 1]$ as follows:
\begin{equation}
\label{eq:multipass:cutoff}
\cutoff_v\paren*{ x } = 
\begin{cases}
	1, &\text{~if~} x \leq - \vYNoAlpha{\dirIn}{v} \\
	\frac{ \vYNoAlpha{\dirOut}{v} - x }{ \vYNoAlpha{\dirIn}{v} + \vYNoAlpha{\dirOut}{v} }, &\text{~if~} {-\vYNoAlpha{\dirIn}{v}} < x \leq \vYNoAlpha{\dirOut}{v} \\
	0, &\text{~if~} \vYNoAlpha{\dirOut}{v} < x
\end{cases} .
\end{equation}
Observe that because of \cref{line:cls17:pos,line:cls17estimate:pos}, it holds that $\aPosEst(v) = \cutoff_v\paren*{ \vZEst{\dirIn}{v} - \vZEst{\dirOut}{v} }$ and $\aPos(v) = \cutoff_v\paren*{ \vZ{\dirIn}{v} - \vZ{\dirOut}{v} }$. Moreover, observe that for all $x, y \in \mathbb{R}$,
\begin{align*}
\abs*{ \cutoff_v(x) - \cutoff_v(y) } &\leq \frac{ \abs*{ x - y } }{ \vYNoAlpha{\dirIn}{v} + \vYNoAlpha{\dirOut}{v} }. \\
\intertext{We can assume that $\card*{ E_{ \dirIn, \dirLow }(v) } + \card*{ E_{ \dirOut, \dirLow }(v) } > 0$ (as the lemma is trivial otherwise), and therefore, by \cref{eq:multipass:y-value} we have:.}
&\leq \frac{ \abs*{ x - y } }{ \epsilon^5 \cdot \paren*{ \card*{ E_{ \dirIn, \dirLow }(v) } + \card*{ E_{ \dirOut, \dirLow }(v) } } } .
\end{align*}
Using this, we derive:
\begin{align*}
&\Pr\paren*{ \abs*{ \aPosEst(v) - \aPos(v) } \geq \delta_{ \chi(v) } } \\
&\hspace{0.5cm}\leq \Pr\paren*{ \abs*{ \cutoff_v\paren*{ \vZEst{\dirIn}{v} - \vZEst{\dirOut}{v} } - \cutoff_v\paren*{ \vZ{\dirIn}{v} - \vZ{\dirOut}{v} } } \geq \delta_{ \chi(v) } } \\
&\hspace{0.5cm}\leq \Pr\paren*{ \frac{ \abs*{ \paren*{ \vZEst{\dirIn}{v} - \vZEst{\dirOut}{v} } - \paren*{ \vZ{\dirIn}{v} - \vZ{\dirOut}{v} } } }{ \epsilon^5 \cdot \paren*{ \card*{ E_{ \dirIn, \dirLow }(v) } + \card*{ E_{ \dirOut, \dirLow }(v) } } } \geq \delta_{ \chi(v) } } \\
&\hspace{0.5cm}= \Pr\paren[\bigg]{ \abs*{ \paren*{ \vZEst{\dirIn}{v} - \vZEst{\dirOut}{v} } - \paren*{ \vZ{\dirIn}{v} - \vZ{\dirOut}{v} } } \geq \delta_{ \chi(v) } \cdot \epsilon^5 \cdot \paren*{ \card*{ E_{ \dirIn, \dirLow }(v) } + \card*{ E_{ \dirOut, \dirLow }(v) } } } \\
&\hspace{0.5cm}\leq \Pr\paren[\bigg]{ \abs*{ \paren*{ \vZEst{\dirIn}{v} - \vZEst{\dirOut}{v} } - \paren*{ \vZ{\dirIn}{v} - \vZ{\dirOut}{v} } } > 10 \delta_{ a - 1 } \cdot \paren*{ \card*{ E_{ \dirIn, \dirLow }(v) } + \card*{ E_{ \dirOut, \dirLow }(v) } } } \tag{\cref{eq:multipass:params}} \\
&\hspace{0.5cm}\leq 4 \delta_{ a - 1 } \tag{\cref{cor:cls17estimate:z}} \\
&\hspace{0.5cm}\leq \delta_{ \chi(v) } ,
\end{align*}
completing the proof of \Cref{lemma:multipass:cls17estimate}.
\end{proof}

Now \Cref{lemma:multipass:cls17estimate} implies:
\begin{lemma}
\label{lemma:multipass:cls17estimate:pair} 
Let $k > 0$, $G = \paren*{ V, E }$ be a graph, and $\chi : V \to [k]$ be a proper coloring of $G$. For all $(u, v) \in E$, we have:
\[
\Pr\paren*{ \abs*{ \aPosEst(u) \cdot \paren*{ 1 - \aPosEst(v) } - \aPos(u) \cdot \paren*{ 1 - \aPos(v) } } \geq 5 \delta_k } \leq 2\delta_k .
\]
\end{lemma}
\begin{proof}
By \cref{lemma:multipass:cls17estimate} and a union bound, we have that except with probability at most $2 \delta_k$, we have $\abs*{ \aPosEst(u) - \aPos(u) }, \abs*{ \aPosEst(v) - \aPos(v) } < \delta_k$. We condition on this event for the rest of this proof and show that it implies that $\abs*{ \aPosEst(u) \cdot \paren*{ 1 - \aPosEst(v) } - \aPos\paren*{ u_j } \cdot \paren*{ 1 - \aPos\paren*{ v_j } } } < 5 \delta_k$.

By our conditioning, $\aPosEst(u) \leq \aPos(u) + \delta_k$ and $\aPosEst(v) \geq \aPos(v) - \delta_k$, and the latter implies $1 - \aPosEst(v) \leq (1-\aPos(v)) + \delta_k$. Since $\aPosEst(u),\aPosEst(v) \in [0,1]$, the LHSes are nonnegative, so we can multiply and deduce:
\begin{align*}
\aPosEst(u) \cdot \paren*{ 1 - \aPosEst(v) } &\leq (\aPos(u) + \delta_k) ((1-\aPos(v)) + \delta_k) \\
&= \aPos(u) (1-\aPos(v)) + \delta_k (1-\aPos(v)) + \delta_k \aPos(u) + \delta_k^2 \\
&\leq \aPos(u) \cdot \paren*{ 1 - \aPos(v) } + 3 \delta_k,
\end{align*}
using again that $\aPos(u), 1-\aPos(v) \leq 1$ and also $\delta_k \leq 1$. We also have:
\begin{align*}
\aPosEst(u) \cdot \paren*{ 1 - \aPosEst(v) } &= \aPosEst(u) - \aPosEst(u) \cdot \aPosEst(v) \\
&\geq \aPos(u) - \delta_k - \paren*{ \aPos(u) + \delta_k } \cdot \paren*{ \aPos(v) + \delta_k } \\
&\geq \aPos(u) - \aPos(u) \cdot \aPos(v) - 4 \delta_k \tag{As $\aPos(u), \aPos(v), \delta_k \leq 1$} \\
&= \aPos(u) \cdot \paren*{ 1 - \aPos(v) } - 4 \delta_k .
\end{align*}
\end{proof}

By our settings of parameters (\Cref{eq:multipass:delta-k-bound}), we have:

\begin{corollary}[Corollary of \cref{lemma:multipass:cls17estimate:pair,claim:cls17exp-helper}]
\label{cor:cls17estimate:pair} 
Let $k > 0$, $G = \paren*{ V, E }$ be a graph, and $\chi : V \to [k]$ be a proper coloring of $G$. For all $(u, v) \in E$, we have:
\[
\abs*{ \Exp\bracket*{ \aPosEst(u) \cdot \paren*{ 1 - \aPosEst(v) } } - \aPos(u) \cdot \paren*{ 1 - \aPos(v) } } \leq 10 \delta_k < \epsilon^5 .
\]
\end{corollary}

\subsection{Estimating the Maximum Directed Cut}
\label{sec:multipass:cls17-dicut}

We are now ready to present our algorithm to estimate the value of the maximum directed cut. As before, we assume that the input graph has a proper coloring $\chi$ with say $k$ colors. As the algorithm is trivial otherwise, we will also assume that the edge set of the input graph is non-empty. Note that, as written, it is not a streaming algorithm. In \cref{sec:multipass:formal}, we explain why it can be implemented in a streaming setting as required for \cref{thm:multipass-formal}.

\begin{algorithm}[H]
	\caption{Estimator for $\maxval{G}$.}
	\label{algo:cls17-dicut}
	\begin{algorithmic}[1]
		\renewcommand{\algorithmicrequire}{\textbf{Input:}}
		\renewcommand{\algorithmicensure}{\textbf{Output:}}
		
		\Require An integer $k > 0$, a graph $G = \paren*{ V, E }$, a proper coloring $\chi : V \to [k]$ of $G$.
		\Ensure An estimate of $\maxval{G}$.

        \alglinenoPop{algcommon}
		
		\State Sample independent $\paren*{ \mathsf{u}_1, \mathsf{v}_1 }, \dots, \paren*{ \mathsf{u}_D, \mathsf{v}_D } \sim \Unif{E}$. Let $\aPosEst(\cdot)$ be as defined in \cref{algo:cls17estimate}. Output the value:
		\[
			\mathsf{Out} \gets \frac{ 1 - 100 \epsilon^5 }{D} \cdot \sum_{j = 1}^D \aPosEst\paren*{ \mathsf{u}_j } \cdot \paren*{ 1 - \aPosEst\paren*{ \mathsf{v}_j } } ,
		\]
		where the randomness used in each invocation of $\aPosEst(\cdot)$ above is independent.

        \alglinenoPush{algcommon}
	\end{algorithmic}
\end{algorithm}

The following lemma states roughly that for any fixed (fractional) assignment $\aX : V \to [0,1]$, the value of $\aX$ on $D$ independent, uniformly sampled edges of $G$ is likely close to the true value of $\aX$ on $G$.

\begin{lemma}
\label{lemma:multipass:cls17-dicut:hoeffding}
For all functions $\aX : V \to [0,1]$, we have:
\[
\Pr_{ \paren*{ \mathsf{u}_1, \mathsf{v}_1 }, \dots, \paren*{ \mathsf{u}_D, \mathsf{v}_D } \sim \Unif{E}^D }\paren*{ \abs*{ \frac{1}{D} \cdot \sum_{j = 1}^D \aX\paren*{ \mathsf{u}_j } \cdot \paren*{ 1 - \aX\paren*{ \mathsf{v}_j } } - \val{G}{\aX} } \geq \epsilon^5 } \leq \epsilon^{20} .
\]
\end{lemma}
\begin{proof}
The proof is a standard Chernoff bound. For all $j \in [D]$, define $\mathsf{X}_j$ be a random variable that takes the value $x\paren*{ \mathsf{u}_j } \cdot \paren*{ 1 - x\paren*{ \mathsf{v}_j } }$. Observe that these random variables are mutually independent and identically distributed with expectation $\frac{1}{ \card*{ E } } \cdot \sum_{ (u, v) \in E } x(u) \cdot \paren*{ 1 - x(v) } = \val{G}{x}$ by \cref{eq:val}. Moreover, they only take values in the interval $[0, 1]$. From \cref{prop:hoeffding,eq:multipass:params}, we have:
\[
\Pr_{ \paren*{ \mathsf{u}_1, \mathsf{v}_1 }, \dots, \paren*{ \mathsf{u}_D, \mathsf{v}_D } \in E }\paren*{ \abs*{ \frac{1}{D} \cdot \sum_{j = 1}^D x\paren*{ \mathsf{u}_j } \cdot \paren*{ 1 - x\paren*{ \mathsf{v}_j } } - \val{G}{x} } \geq \epsilon^5 } \leq 2 \cdot \mathrm{e}^{ - 2D \cdot \epsilon^{10} } < \epsilon^{20} .
\]
\end{proof}

Now, we apply the previous lemma in the case $\aX = \aPos$, along with \cref{cor:cls17estimate:pair} which states that $\aPos$ and $\aPosEst$ likely give nearby values on any fixed edge, to deduce overall correctness of our estimate:

\begin{lemma}
\label{lemma:multipass:cls17-dicut}
Let $k > 0$, $G = \paren*{ V, E }$ be a graph, and $\chi : V \to [k]$ be a proper coloring of $G$. We have:
\[
\Pr\paren*{ \paren*{ \frac{1}{2} - 200 \epsilon^5 } \cdot \maxval{G} \leq \mathsf{Out} \leq \maxval{G} } \geq 1 - \epsilon^{10} .
\]
\end{lemma}
\begin{proof}
Recall from \cref{lemma:multipass:cls17} that $\paren*{ \frac{1}{2} - \epsilon^5 } \cdot \maxval{G} \leq \val{G}{ \aPos } \leq \maxval{G}$. It follows that:
\[
\paren*{ \frac{1}{2} - 100 \epsilon^5 } \cdot \maxval{G} \leq \paren*{ 1 - 100 \epsilon^5 } \cdot \val{G}{ \aPos } \leq \paren*{ 1 - 100 \epsilon^5 } \cdot \maxval{G} .
\]
This means that the lemma we wish to bound follows if we show that:
\[
\Pr\paren*{ \abs*{ \mathsf{Out} - \paren*{ 1 - 100 \epsilon^5 } \cdot \val{G}{ \aPos } } \geq 100 \epsilon^5 \cdot \maxval{G} } \leq \epsilon^{10} .
\]
By definition of $\mathsf{Out}$ and the fact that $\maxval{G} \geq \frac{1}{4}$, this follows if we show that:
\[
\Pr_{ \paren*{ \mathsf{u}_1, \mathsf{v}_1 }, \dots, \paren*{ \mathsf{u}_D, \mathsf{v}_D } \sim E }\paren*{ \abs*{ \frac{1}{D} \cdot \sum_{j = 1}^D \aPosEst\paren*{ \mathsf{u}_j } \cdot \paren*{ 1 - \aPosEst\paren*{ \mathsf{v}_j } } - \val{G}{ \aPos } } \geq 20 \epsilon^5 } \leq \epsilon^{10} .
\]
Let $\mathcal{E}$ be the event we upper bounded in \cref{lemma:multipass:cls17-dicut:hoeffding} for $\aX = \aPos$. By a chain rule, we have:
\begin{multline*}
    \Pr_{ \paren*{ \mathsf{u}_1, \mathsf{v}_1 }, \dots, \paren*{ \mathsf{u}_D, \mathsf{v}_D } \sim E }\paren*{ \abs*{ \frac{1}{D} \cdot \sum_{j = 1}^D \aPosEst\paren*{ \mathsf{u}_j } \cdot \paren*{ 1 - \aPosEst\paren*{ \mathsf{v}_j } } - \val{G}{ \aPos } } \geq 20 \epsilon^5 } \\
    \leq \Pr\paren*{ \mathcal{E} } + \Pr_{ \paren*{ \mathsf{u}_1, \mathsf{v}_1 }, \dots, \paren*{ \mathsf{u}_D, \mathsf{v}_D } \sim E }\paren*{ \abs*{ \frac{1}{D} \cdot \sum_{j = 1}^D \aPosEst\paren*{ \mathsf{u}_j } \cdot \paren*{ 1 - \aPosEst\paren*{ \mathsf{v}_j } } - \val{G}{ \aPos } } \geq 20 \epsilon^5 ~\bigg\vert~ \overline{ \mathcal{E} } }.
\end{multline*}
By \cref{lemma:multipass:cls17-dicut:hoeffding}, it suffices to bound the second term by $\epsilon^{20}$. We shall show this bound under a stronger conditioning by fixing an arbitrary $\paren*{ u_1, v_1 }, \dots, \paren*{ u_D, v_D } \in E$ for which $\mathcal{E}$ does not occur. Observe that, after this stronger conditioning, the only randomness left is the randomness in the computations of $\aPosEst(\cdot)$. For an arbitrary such $\paren*{ u_1, v_1 }, \dots, \paren*{ u_D, v_D } \in E$, we have:
\begin{align*}
    &\Pr\paren*{ \abs*{ \frac{1}{D} \cdot \sum_{j = 1}^D \aPosEst\paren*{ \mathsf{u}_j } \cdot \paren*{ 1 - \aPosEst\paren*{ \mathsf{v}_j } } - \val{G}{ \aPos } } \geq 20 \epsilon^5 ~\bigg\vert~ \paren*{ u_1, v_1 }, \dots, \paren*{ u_D, v_D } } \\
    &\leq \Pr\paren*{ \abs*{ \frac{1}{D} \cdot \sum_{j = 1}^D \aPosEst\paren*{ u_j } \cdot \paren*{ 1 - \aPosEst\paren*{ v_j } } - \val{G}{ \aPos } } \geq 20 \epsilon^5 } ,
\intertext{as the randomness in the computations of $\aPosEst(\cdot)$ is independent of the randomness used to sample $\paren*{ \mathsf{u}_1, \mathsf{v}_1 }, \dots, \paren*{ \mathsf{u}_D, \mathsf{v}_D } \sim E$. Next, by our choice of $\paren*{ u_1, v_1 }, \dots, \paren*{ u_D, v_D }$, we have that $\abs*{ \frac{1}{D} \cdot \sum_{j = 1}^D \aPos\paren*{ u_j } \cdot \paren*{ 1 - \aPos\paren*{ v_j } } - \val{G}{ \aPos } } \leq \epsilon^5$. By a triangle inequality, this means that:}
&\leq \Pr\paren*{ \abs*{ \frac{1}{D} \cdot \sum_{j = 1}^D \aPosEst\paren*{ u_j } \cdot \paren*{ 1 - \aPosEst\paren*{ v_j } } - \frac{1}{D} \cdot \sum_{j = 1}^D \aPos\paren*{ u_j } \cdot \paren*{ 1 - \aPos\paren*{ v_j } } } \geq 18 \epsilon^5 } .
\intertext{We now apply \cref{cor:cls17estimate:pair} on the edges $\paren*{ u_1, v_1 }, \dots, \paren*{ u_D, v_D }$ to get that for all $j \in [D]$, we have the bound $\abs*{ \Exp\bracket*{ \aPosEst\paren*{ u_j } \cdot \paren*{ 1 - \aPosEst\paren*{ v_j } } } - \aPos\paren*{ u_j } \cdot \paren*{ 1 - \aPos\paren*{ v_j } } } \leq \epsilon^5$. By a triangle inequality, this means that we have $\abs*{ \frac{1}{D} \cdot \sum_{j = 1}^D \Exp\bracket*{ \aPosEst\paren*{ u_j } \cdot \paren*{ 1 - \aPosEst\paren*{ v_j } } } - \frac{1}{D} \cdot \sum_{j = 1}^D \aPos\paren*{ u_j } \cdot \paren*{ 1 - \aPos\paren*{ v_j } } } \leq \epsilon^5$. By another triangle inequality, we get:}
&\leq \Pr\paren*{ \abs*{ \frac{1}{D} \cdot \sum_{j = 1}^D \aPosEst\paren*{ u_j } \cdot \paren*{ 1 - \aPosEst\paren*{ v_j } } - \frac{1}{D} \cdot \sum_{j = 1}^D \Exp\bracket*{ \aPosEst\paren*{ u_j } \cdot \paren*{ 1 - \aPosEst\paren*{ v_j } } } } \geq 5 \epsilon^5 } .
\intertext{We now apply Hoeffding's inequality (\cref{prop:hoeffding}). We get using \Cref{eq:multipass:exp-est} that:}
&\leq 2 \cdot \mathrm{e}^{ - 10 D \cdot \epsilon^{10} } < \epsilon^{20},
\end{align*}
as desired.
\end{proof}

\subsection{Streaming implementation of \cref{algo:cls17estimate}}
\label{sec:multipass:cls17estimate-streaming}

\cref{lemma:multipass:cls17-dicut} shows that \cref{algo:cls17-dicut} on input $G$ and $\chi$ computes an estimate of $\maxval{G}$ with high probability. However, it is not in the streaming model and therefore not enough for \cref{thm:multipass-formal}. We now show that it can be implemented in the streaming setting. For this, we first show how \cref{algo:cls17estimate}, the main building block of \cref{algo:cls17-dicut} is implementable in the streaming setting. Namely, we show that:

\begin{lemma}
\label{lemma:multipass:cls17estimate-streaming}
For all $\epsilon > 0$ and integers $0 < a \leq k$, there exists a randomized streaming algorithm that on input an $n$-vertex directed graph $G = \paren*{ V, E }$, a proper coloring $\chi : V \to [k]$ of $G$ using $k$ colors, and a non-isolated vertex $v \in V$ satisfying $\chi(v) \leq a$, uses $a$ passes and $\paren*{ 10 D }^a \cdot \log n$ space and outputs a value identically distributed to $\aPosEst(v)$.
\end{lemma}
\begin{proof}
Proof by induction on $a$. For the base case $a = 1$, observe from \cref{algo:cls17estimate,eq:multipass:edge-parts} that for all $v \in V$ satisfying $\chi(v) = 1$, we have $\aPosEst(v) = \frac{ \vYNoAlpha{\dirOut}{v} }{ \vYNoAlpha{\dirIn}{v} + \vYNoAlpha{\dirOut}{v} } = \frac{ \card*{ E_{ \dirOut, \dirHigh }(v) } }{ \card*{ E_{ \dirIn, \dirHigh }(v) } + \card*{ E_{ \dirOut, \dirHigh }(v) } }$ (with probability $1$), where $\card*{ E_{ \dirIn, \dirHigh }(v) }$ and $\card*{ E_{ \dirOut, \dirHigh }(v) }$ are the in-degree and out-degree of $v$ respectively. These can be easily computed in one pass and $10 \log n$ space, and the lemma follows. 

We now show the lemma for $a > 1$ assuming it holds for $a - 1$. Fix $\epsilon$ and $k$ as in the lemma and let $\mathbb{A}_{a - 1}$ be the streaming algorithm promised by the induction hypothesis. We will need the following well known lemma, that follows from reservoir sampling of \cite{Vit85}. Let $D = \epsilon^{ - 100k }$ be as in \cref{eq:multipass:params}.
\begin{lemma}[Corollary of \cite{Vit85}]
\label{lemma:multipass:vit85}
There exists a randomized streaming algorithm $\mathbb{A}_{ \mathsf{Samp} }$ that on input an $n$-vertex directed graph $G = \paren*{ V, E }$, a proper coloring $\chi : V \to [k]$ of $G$ using $k$ colors, and a vertex $v \in V$ uses $1$ pass and $2D \log n$ space and outputs $D$ independent and uniformly random samples from $E_{ \dirIn, \dirLow }(v)$ if $E_{ \dirIn, \dirLow }(v) = \emptyset$, and otherwise outputs $\bot$. An analogous result holds for any union of the multisets defined in \cref{eq:multipass:edge-parts}.
\end{lemma}

To show the lemma for $a$, we define the algorithm $\mathbb{A}_a$ that imitates \cref{algo:cls17estimate}. This is done by first using one pass to run \cref{lemma:multipass:vit85} and compute the samples and values required in \cref{line:cls17estimate:z}. These samples, together with the algorithm $\mathbb{A}_{a - 1}$ are then used to compute $\vZEst{\dirIn}{v}$ and $\vZEst{\dirOut}{v}$, which is then used to compute the value $\aPosEst(v)$ exactly as in \cref{line:cls17estimate:pos}. The correctness and the space and pass complexity of $\mathbb{A}_a$ are straightforward from the induction hypothesis. 
\end{proof}

\subsection{Proof of \cref{thm:multipass-formal}}
\label{sec:multipass:formal}

We are now ready to prove \cref{thm:multipass-formal}. Fix $\epsilon > 0$ and define $k = \frac{200}{ \epsilon } - 1 \geq \frac{150}{ \epsilon }$ as we assume $\epsilon < 0.01$. Sample a $2$-wise independent hash function $\chi : V \to [k]$ and define the algorithm:

\begin{algorithm}[H]
	\caption{The algorithm $\mathbb{A}_{ \mathsf{multipass} }$ proving \cref{thm:multipass-formal}.}
	\label{algo:multipass-dicut}
	\begin{algorithmic}[1]
		\renewcommand{\algorithmicrequire}{\textbf{Input:}}
		\renewcommand{\algorithmicensure}{\textbf{Output:}}
		
		\Require A graph $G = \paren*{ V, E }$. As the stream consists of edges, we can without loss of generality that $G$ has no isolated vertices and at least one edge.
		\Ensure A value $\cutest$. 
		\alglinenoPop{algcommon}
		\State Throughout this algorithm, we ignore all edges $(u, v)$ in the stream for which $\chi(u) = \chi(v)$. This is the same as saying what follows is run on a graph without these edges. By definition, $\chi$ is a proper coloring for this remaining graph $G' = \paren*{ V', E' }$ with $k$ colors. \label{line:multipass-dicut:crlot}

		\State If $G'$ has no edges, output $0$. Otherwise, run the algorithm $\mathbb{A}_{ \mathsf{Samp} }$ from \cref{lemma:multipass:vit85} to get samples $\paren*{ \mathsf{u}_1, \mathsf{v}_1 }, \dots, \paren*{ \mathsf{u}_D, \mathsf{v}_D } \in E'$. This takes one pass and space $10 D \log n$. \label{line:multipass-dicut:samp}
		
		\State Output the value: \label{line:multipass-dicut:out}
		\[
			\cutest \gets \frac{ 1 - 100 \epsilon^5 }{D} \cdot \sum_{j = 1}^D \mathbb{A}_k\paren*{ \mathsf{u}_j } \cdot \paren*{ 1 - \mathbb{A}_k\paren*{ \mathsf{v}_j } } ,
		\]
		where $\mathbb{A}_k$ is as promised \cref{lemma:multipass:cls17estimate-streaming} for $a = k$ and the executions of $\mathbb{A}_k$ are done in parallel using independent randomness.	
        \alglinenoPush{algcommon}
    \end{algorithmic}
\end{algorithm}

\begin{proof}[Proof of \cref{thm:multipass-formal}]
The space and pass complexity are straightforward and we simply prove correctness. Let $\mathcal{E}_1$ be the event that $\card*{ E' } < \paren*{ 1 - \frac{ \epsilon }{8} } \cdot \card*{ E }$. Due to \cref{prop:coloring}, we have that $\Pr\paren*{ \mathcal{E}_1 } \leq \frac{1}{20}$. To show the theorem, it suffices to condition on $\overline{ \mathcal{E}_1 }$ and show that the required probability is at least $\frac{19}{20}$. We shall in fact show this under a stronger conditioning by conditioning on an arbitrary graph $G'$ for which $\mathcal{E}_1$ does not occur. Observe that the fact that $G$ has at least one edge implies that $G'$ also has at least one edge and \cref{line:multipass-dicut:samp} does not output $0$. 

As \cref{line:multipass-dicut:samp} does not output $0$, we can conclude from \cref{lemma:multipass:vit85,lemma:multipass:cls17estimate-streaming} that \cref{algo:multipass-dicut} simply imitates \cref{algo:cls17-dicut} on the graph $G'$ that has proper coloring $\chi$. This means that:
\begin{align*}
&\Pr\paren*{ \paren*{ \frac{1}{2} - \epsilon } \cdot \maxval{G} \leq \cutest \leq \maxval{G} \mid G' } \\
&\hspace{1cm}= \Pr\paren*{ \paren*{ \frac{1}{2} - \epsilon } \cdot \maxval{G} \leq \mathsf{Out}\paren*{ G' } \leq \maxval{G} \mid G' } \\
&\hspace{1cm}= \Pr\paren*{ \paren*{ \frac{1}{2} - \epsilon } \cdot \maxval{G} \leq \mathsf{Out}\paren*{ G' } \leq \maxval{G} } \tag{$\mathsf{Out}$ is independent of $G'$} \\
&\hspace{1cm}\geq \Pr\paren*{ \paren*{ \frac{1}{2} - \epsilon^2 } \cdot \maxval{ G' } \leq \mathsf{Out}\paren*{ G' } \leq \maxval{ G' } } \tag{$G'$ is a subgraph of $G'$ and $\mathcal{E}_1$ does not happen} \\
&\hspace{1cm}\geq \frac{19}{20} . \tag{\cref{lemma:multipass:cls17-dicut}}
\end{align*}
\end{proof}

\ifnum\doubleblind=0

\section*{Acknowledgments}

\textsc{r.r.s.} is supported by the Department of Atomic Energy, Government of India, under project no. RTI4001.
\textsc{n.g.s.} is supported in part by an NSF Graduate Research Fellowship (Award DGE 2140739).
\textsc{m.s.} is supported in part by a Simons Investigator Award and NSF Award CCF 2152413.
\textsc{s.v.} is supported in part by NSF award CCF 2348475. Part of the work was conducted when \textsc{s.v.} was visiting the Simons Institute for the Theory of Computing as a research fellow in the Sublinear Algorithms program. Part of the work was conducted when \textsc{s.v.} was a graduate student at Harvard University, and supported in part by a Google Ph.D. Fellowship, a Simons Investigator Award to \textsc{m.s.}, and NSF Award CCF 2152413.

We would like to thank Hoaian Nguyen and the anonymous reviewers at SODA for helpful comments on the text.
\fi
\printbibliography

\end{document}